\documentclass[11pt]{article}
\usepackage{fullpage}
\usepackage{authblk}

\usepackage{ifthen}
\usepackage{graphics,epsfig}
\usepackage{amsmath,amsfonts,amssymb}
\usepackage{algorithm}
\usepackage{algpseudocode}
\usepackage{color}
\usepackage{stmaryrd}
\usepackage{tabulary}

\newtheorem{theorem}{Theorem}
\newtheorem{lemma}{Lemma}

\newtheorem{proposition}{Proposition}
\newtheorem{claim}{\rm\underline{Claim}}
\newtheorem{fact}{Fact}
\newtheorem{definition}{Definition}
\newcommand{\qed}{\hfill $\Box$ \medbreak}
\newenvironment{proof}{\noindent {\bf Proof.}}{\qed}

\newcommand{\val}{\textsl{val}}

\newcommand{\ball}{\textsl{ball}}
\newcommand{\good}{\textsl{good}}
\newcommand{\actions}{\textsl{actions}}
\renewcommand{\time}{\textsl{time}}
\newcommand{\pref}{\textsl{pref}}
\newcommand{\cD}{{\cal D}}
\newcommand{\cF}{{\cal F}}
\newcommand{\cL}{{\cal L}}

\begin{document}

\title{Local Distributed Algorithms for Selfish Agents\footnote{All authors received supports from the ANR project DISPLEXITY, and from the INRIA project GANG.}}

\author[1]{Simon Collet\thanks{Funded by the European Research Council (ERC) under the H2020 research and innovation program (grant  No 648032).}}
\author[1]{Pierre Fraigniaud}
\author[2]{Paolo Penna}
  
\affil[1]{\small CNRS and University Paris Diderot, France.}
\affil[2]{\small ETH Zurich, Switzerland.}

\date{}

\maketitle

\begin{abstract}
In the classical framework of \emph{local distributed network computing}, it is generally assumed that the entities executing distributed algorithms are \emph{altruistic}. However, in various scenarios, the value of the output produced by an entity may have a tremendous impact on its future. This is for instance the case of tasks such as computing maximal independent sets (MIS) in networks. Indeed, a node belonging to the MIS may be later asked more than to a node not in the MIS --- e.g., because MIS in networks are often used as backbones to collect, transfer, and broadcast information, which is costly. In this paper, we revisit typical local distributed network computing tasks in the framework of algorithmic game theory. Specifically, we focus on the construction of solutions for \emph{locally checkable labeling} (LCL) tasks, which form a large class of distributed tasks, including MIS, coloring, maximal matching, etc.,  and which have been studied for more than 20 years in distributed computing. 

Given an LCL task, the nodes are collectively aiming at computing a solution, but, at the individual level, every node plays \emph{rationally} and  \emph{selfishly} with the objective  of optimizing its own profit. Surprisingly, the classical frameworks for game theory are not fully appropriate for the purpose of our study. Moreover, we show that classical notions like \emph{Nash equilibria} may yield algorithms requiring an arbitrarily large number of rounds to converge. Nevertheless, by extending and generalizing core results from game theory, we establish the existence of  a so-called \emph{trembling-hand perfect equilibria}, a subset  of Nash equilibria that is well suited to LCL construction tasks. The main outcome of the paper is therefore that, for essentially all distributed tasks whose solutions are locally checkable, there exist  construction algorithms which are robust to selfishness, i.e., which are guaranteeing  that nodes have no incentive to deviate from the actions instructed by these algorithms.  
\end{abstract}

\vfill

\thispagestyle{empty}
\setcounter{page}{0}
\newpage

%%%%%%%%%%%%%%%%%%%%%%%%%%%%%%%%%%%%%%%%%%%%%
\section{Introduction}
%%%%%%%%%%%%%%%%%%%%%%%%%%%%%%%%%%%%%%%%%%%%%

%---------------------------------------------------------------------------
\subsection{Context and Objective}
%---------------------------------------------------------------------------

 In the conceptual framework of local distributed computing in networks~\cite{Peleg00}, algorithms solving tasks such as minimum-weight spanning tree, maximum matching, coloring, maximal independent set, etc., usually assume that every node of the network will obey their instructions.  This might be appropriate in a context in which each node is altruistic and/or oblivious to its output. However, there exist several settings in which the future of a node depends on its output. For instance, independent sets and dominating sets  can be used as backbones to collect, transfer, and broadcast information, and/or as cluster heads in clustering protocols~\cite{KuhnMW04,MoscibrodaW05a}. Hence, belonging to the independent set or to the dominating set may generate future costs in term of  energy consumption, computational efforts, and bandwidth usage. As a consequence, rational selfish nodes might be tempted to deviate from the instructions of the algorithm, so that to avoid becoming member of the independent set or dominating set under construction. On the other hand, the absence of a backbone or of cluster heads may penalize the nodes. Hence every node is subject to a tension between (1)~facilitating the construction of a solution, and (2) avoiding certain types of solutions. The objective of this paper is to design algorithms enabling to resolve this tension. 
 
For example, a large class of randomized maximal independent set (MIS) algorithms~\cite{AlonBI86,Luby86} proceed in rounds, where a round allows every node to exchange information with its neighbors in the network. Roughly, at each round, every node~$i$ which has not yet decided proposes to enter the MIS with a certain probability $p_i$. If a node proposes itself to enter the MIS, and none of its neighbors is also proposing itself to enter the MIS, then it enters the MIS. If two adjacent nodes simultaneously propose themselves to enter the MIS, then there is a conflict, and both nodes remain undecided. The quality of the algorithm heavily depends on the choice of the probability $p_i$ that node $i$ proposes itself to enter the MIS, which may typically depend on the degree of  node~$i$, and may vary along with the execution of the algorithm as node~$i$ accumulates more and more information about its neighborhood. Hence, a node~$i$ aiming at avoiding entering the MIS might be tempted to deviate from the algorithm by setting $p_i$ small. However, if all nodes deviate in this way, then the algorithm may take a very long time before converging to a solution.
 
Coloring is another example where nodes may adopt selfish strategies to avoid some specific solutions. This is for instance the case where coloring is viewed as an abstraction of frequency assignment in radio networks. Some frequencies might be preferred to others because, e.g.,  some frequencies might be conflicting with local devices present at the node. As a consequence, not all colors are equal for the nodes, and while each node is aiming at constructing a coloring quickly (in order to take benefit from the resulting radio network), it is also aiming at avoiding being assigned a color that does not please it. Now, similarly to MIS construction, a large class of randomized coloring algorithms proceed in rounds, roughly as follows (see, e.g.,~\cite{BarenboimE13}). At each round, every node that has not yet decided a color (1)~picks a color uniformly at random among the available colors (i.e., among a given palette, but the colors already decided by neighboring nodes), (2)~checks whether this color is conflicting with the random colors picked by the neighbors, and (3)~ decides this color in absence of conflicts. If a node has preferences to some colors, then it might be tempted to give more weight to its preferred colors in its random choices. Again, if all nodes deviate this way, then the algorithm may take a long time before converging to a solution, if converging at all.  

In this paper, we address the issue of designing distributed algorithms that are robust to  selfish behaviors of  the nodes. Obviously, algorithmic game theory has been conceived especially for tackling such  issues. So, more precisely, we are aiming at answering the following question: 

%\medskip

%\noindent\textbf
\paragraph{Existence of  algorithms coping with non-cooperative behavior:} Do there exist some forms of equilibria in the games formalizing distributed tasks such as the one sketched above? That is, do there exist strategies, i.e., probability distributions over the random choices to be made by each node at each round, from which no rational selfish players have incentive to deviate? More importantly, if such equilibria exist, are they corresponding to efficient algorithms for solving the tasks? 

%---------------------------------------------------------------------------
\subsection{Our Results}
%---------------------------------------------------------------------------

First, we define \emph{LCL games}, induced by a generic form of randomized algorithms for constructing solutions to \emph{locally checkable labeling} (LCL) tasks (cf. Section~\ref{sec:lcldef}). Recall that locally checkable labeling tasks~\cite{NaorS95} form a large class of classical problems, including maximal independent set, coloring, maximal matching, minimal dominating set, etc.,  studied for more than 20 years in the framework of distributed computing.

Next, we observe that classical frameworks for algorithmic game theory do not fully capture the essence of distributed games (cf. Section~\ref{sec:rw}). Moreover, we show that the classical notion of Nash equilibria is not sufficient to characterize desirable distributed algorithms, i.e., efficient and robust to the selfish behaviors of the nodes. In particular, we show that there are Nash equilibria which correspond to algorithms whose running times, i.e., number of rounds, are arbitrarily large (cf. Proposition~\ref{theo:nashisbad}). 

In contrast, we show that the stronger notion of \emph{trembling-hand perfect equilibria} (a subset of Nash equilibria) does not suffer from the above problem. In fact, trembling-hand perfect equilibrium is a notion of equilibrium that is well suited to the purpose of designing efficient and robust randomized distributed algorithms for LCL tasks. Indeed, trembling-hand perfect equilibria are limits of Nash equilibria of perturbed games, that is, games in which players must play completely mixed strategies (i.e. strategies in which each action has a constant probability of being played). As such, trembling-hand perfect equilibria do not suffer from the undesirable phenomenons of players blocking each other for an arbitrarily long time. Indeed, this phenomenon results from weakly dominated strategies that stick to Nash equilibria in which players repeat the same choices again and again for an arbitrarily long amount of time. 

We prove that all LCL games whose solution can be constructed by a greedy sequential algorithm, such as MIS, $(\Delta+1)$-coloring, maximal matching, etc., have trembling-hand perfect equilibria. Specifically, our main result is the following: 

\begin{theorem}\label{cor:mainMIS}
Let $\cL$ be a greedily constructible locally checkable labeling. The game associated to~$\cL$ has a symmetric trembling-hand perfect equilibrium. 
\end{theorem}

This result is achieved by establishing a general result, of independent interest, stated below as a lemma: 

\begin{lemma}\label{theo:maingeneral}
Every infinite, continuous, measurable, well-rounded, symmetric extensive game with perfect recall and finite action set has a symmetric trembling-hand perfect equilibrium.
\end{lemma}

Our results open a new line of research aiming at designing distributed algorithms that are not only efficient, but also robust to  selfish behaviors of the nodes. Indeed, Theorem~\ref{cor:mainMIS} establishes that such algorithms \emph{exist}, for all ``natural'' locally checkable labeling (LCL) tasks. 

%---------------------------------------------------------------------------
\subsection{Additional Related Work}
%---------------------------------------------------------------------------

The  readers interested in specific important distributed network computing problems such as MIS and $(\Delta+1)$-coloring  are invited to consult~\cite{FHK16,Gha16,HSS16} and the references therein for the most recent advances in the field. See also~\cite{Suo13} for a survey, and~\cite{CKP16,FeuilloleyF15,NaorS95} for results on derandomization techniques for LCL tasks.

Other forms of game theoretic approaches have been studied in the past, in the context of distributed computing, and/or using a distributed computing approach for solving graph problems. In~\cite{EW11}, the authors study programming incentives to the programmers for writing codes that are efficient for the overall system while optimizing individual program’s performances, in the context of transactional memory. In~\cite{AugerCCR13}, the authors describe and analyze ways of computing max-cuts by interpreting these cuts as the pure Nash equilibria of an $n$-player non-zero sum game, each vertex being an agent trying to maximize her selfish interest. 

The papers the most closely related to our work are~\cite{Nis99,NR99,NSVZ11}. In the survey~\cite{Nis99}, the author addresses several problems related to distributed computing, under a large spectrum of different models. It is shown that standard mechanism design techniques (e.g., Vickrey-Groves-Clark mechanism) can be applied to some of these problems. However, such an application requires a central authority which knows everything about the system, except of course the players private information. On the other hand, the paper also describes and analyzes a two-step algorithm for solving MIS on paths, where the players are bounded to communicate only with their neighbors. The algorithms however requires a trusted entity that distributes payments to the players. The class of best-response mechanisms is studied in~\cite{NSVZ11}. Such mechanisms apply to strategic games with a unique equilibrium, in which the players proceed by repeatedly best-respond to each other. Several important problems can be casted in this framework, including the BGP protocol, TCP-like games, among others. However, for applying the mechanism, the underlying game must have a unique equilibrium (plus other technical conditions), and the players are oblivious about the time required to compute the solution. Even the first condition alone may not be satisfied by LCL games (it is typically not satisfied by the MIS game). 

In contrast to the above, our setting does not assume any central entity nor payments, it applies to general graphs and to all ``natural'' LCL tasks (including MIS). We prove the existence of suited algorithms in this general setting. 

%%%%%%%%%%%%%%%%%%%%%%%%%%%%%%%%%%%%%%%%%%%%%
\section{LCL Games}
\label{sec:lcldef}
%%%%%%%%%%%%%%%%%%%%%%%%%%%%%%%%%%%%%%%%%%%%%

%---------------------------------------------------------------------------
\subsection{LCL languages}
%---------------------------------------------------------------------------

Our framework is distributed network computing, for which we adopt the following standard setting (see, e.g., \cite{FraigniaudKP13}). Let $A$ be a finite alphabet. A network \emph{configuration} is a pair $(G,s)$ where $G=(V,E)$ is a simple connected graph, and $s:V\to A$ is a state function, assigning a state $s(v)\in A$ to every node $v\in V$. A \emph{distributed language} over a graph family $\cF$ is a  Turing computable set $\cL$ of configurations such that, for every graph $G\in \cF$, there is a state function~$s$ for which $(G,s)\in \cL$. The class LCL, for \emph{locally checkable labeling}, has been defined in~\cite{NaorS95}. Let $\Delta\geq 2$, and let $\cF_\Delta$ be the family of simple connected graphs with maximum degree~$\Delta$. An LCL language over $\cF_\Delta$ is a distributed language defined by a set of balls of radius at most~$t$ in graphs belonging to $\cF_\Delta$, for some $t\geq 0$, with nodes labeled by values in~$A$. These balls are called the \emph{good} balls of $\cL$.  The set of good balls of an LCL language $\cL$ is denoted by $\good(\cL)$. The parameter $t$ is called the \emph{radius} of $\cL$. Throughout the paper, we shall solely consider graphs with degrees upper bounded by some $\Delta\geq 2$, and thus we ignore the index $\Delta$ in $\cF_\Delta$ in the following. 

For instance, MIS is the LCL language of radius~1 with $A=\{0,1\}$,  where $\good(\cL)$ is the set of balls with radius~1 centered at a node $v$, such that either $s(v)=0$ and there exists a neighbor~$w$ of~$v$ with $s(w)=1$, or $s(v)=1$ and $s(w)=0$ for every neighbor $w$ of $v$. Similarly, $(\Delta+1)$-coloring is the LCL language of radius~1 with $A=\{1,\dots,\Delta+1\}$,  where $\good(\cL)$ is the set of balls with radius~1 whose center $v$ satisfies $s(v)\neq s(w)$ for every neighbor $w$ of $v$. 

Every LCL language $\cL$ induces an associated construction \emph{task}. Given $G\in \cF$, each of its nodes~$v$ has to  compute a value $s(v)\in A$ such that $(G,s)\in\cL$. In order to guarantee the existence of efficient distributed algorithms for solving the task induced by $\cL$, we consider a subclass of LCL languages, called  \emph{greedily constructible}. Roughly, an LCL language $\cL$ is greedily constructible if any partial solution can be extended. More formally, let us consider an extra value $\bot$, whose interpretation is that a node $v$ with $s(v)=\bot$ has not yet decided an output value in~$A$. For a language $\cL$ of radius~$t$, a radius-$t$ ball $B$ is said to be \emph{partially} good if $B$ is resulting from a good ball by replacing the value of some of its node labels by~$\bot$. Let $G=(V,E) \in \cF$, and let $s:V\to A\cup\{\bot\}$ for which all balls in $(G,s)$ are partially good. Let $U=\{v\in V : s(v)=\bot\}$, and let $s':V\to A$ be such that $s'(v)=s(v)$ for all nodes $v\in V\setminus U$. We call such a function~$s'$ an  \emph{extension} of~$s$. Let $U'$ be composed of all nodes in $U$ whose balls in $(G,s')$ are good, and let $s'':V\to A\cup\{\bot\}$ be defined as $s''(v)=s(v)$ for $v\in V\setminus U$, $s''(v)=\bot$ for $v\in U \setminus U'$, and $s''(v)=s'(v)$ for $v\in U'$. An LCL language $\cL$ is greedily constructible if for every such $G$ and $s$ such that $U\neq \emptyset$, we have (1) there exists an extension $s'$ of $s$ such that $s''\neq s$, and, (2) for every extension $s'$ of $s$ such that $s''\neq s$, all balls in $(G,s'')$ are partially good. For instance, $k$-coloring is greedily constructible for any $k>\Delta$, while $\Delta$-coloring is not (even in $\Delta$-colorable graphs). MIS and maximal matching also fall into this class as long as we enforce that nodes (resp., edges) cannot decide by themselves  not to enter the MIS (resp., matching), but can only do so \emph{after} a neighbor has entered the MIS  (resp.,  matching). 

%---------------------------------------------------------------------------
\subsection{A generic randomized algorithm for LCL languages}
%---------------------------------------------------------------------------

The generic randomized construction algorithm for LCL languages described in this subsection is conceived for the classical {\sc local} model of distributed computation~\cite{Peleg00}. Recall that this model assumes that all nodes start simultaneously, and perform synchronously in a sequence of rounds.  At each round, every node exchanges information with its neighbors, and performs individual computations. This goes on at every node until the node decides some output value, and terminates. The nodes do not need to terminate at the same round. However, the output value of a terminated node is irrevocable. Our generic algorithm is described in Algorithm~\ref{alg:typical-form}. 

\begin{algorithm}[htb]
  \caption{\sl Generic algorithm executed at a node $v$.}
  \label{alg:typical-form}
  \small
  \begin{algorithmic}[1]
  \State{$\val(v) \leftarrow \bot$}
  \State{observe $\ball(v)$} \label{ins:obs1}
  \While{$\ball(v)\notin \good(\cL)$} \label{ins:while}
  \State{set $\val(v) \in A$ compatible with $\ball(v)$, at random, according to a probability distribution $\cD$} \label{ins:rand}
  \State{observe $\ball(v)$}   \label{ins:obs2}
      \EndWhile
  \State{\Return$\val(v)$}
  \end{algorithmic}
\end{algorithm}

In Algorithm~\ref{alg:typical-form}, every node~$v$ aims at computing $\val(v)\in A$. It starts with initial value $\val(v)=\bot$. At Step~\ref{ins:obs1} node~$v$ collects the $t$-ball $\ball(v)$ around it, i.e., the values of all the labels of nodes at distance at most $t$ from it, including the structure of the ball, where $t$ is the radius of~$\cL$. (This can be implemented in $t$ rounds in the {\sc local} model). At Step~\ref{ins:while}, node $v$ checks whether the observed ball, $\ball(v)$, is good. If not, then Step~\ref{ins:rand} consists for~$v$ to select a random value $\val(v)\in A$ compatible with $\ball(v)$, i.e., with values of nodes in the ball fixed at previous rounds. The random choice of $\val(v)$ is governed by a probability distribution~$\cD$, which may depend on the round, and on the structure of $v$'s neighborhood as observed during the previous rounds. In the case of MIS, this choice essentially consists in choosing  whether $v$ proposes itself for entering the MIS, with probability $p_v$. Luby's algorithm~\cite{Luby86} sets $p_v = \frac{1}{2\deg(v)}$ where $\deg(v)$ is the degree of $v$ in the graph, not counting the neighbors which already decided at previous rounds. Similarly, in the case of $(\Delta+1)$-coloring, the  choice of $\val(v)$ essentially consists in proposing a color at random. Barenboim-Elkin's algorithm~\cite{BarenboimE13} sets the probability $p_c$ of proposing color $c$ as $\frac{1}{2}$ for $c=0$ (a ``fake'' color), and $p_c=\frac{1}{2k}$ for any of the $k$ colors $c\in\{1,\dots,\Delta+1\}$ that have not yet been decided by some neighboring node(s) during the previous rounds. Once the new value $\val(v)$ has been selected, node~$v$ observes $\ball(v)$ in Step~\ref{ins:obs2}, including all selected values of nodes in this ball. If the selected values by the nodes in this ball make the ball good, then $v$ terminates and outputs $\val(v)$. Otherwise, $v$ starts a new iteration of the while-loop. 

For any greedily constructible language $\cL$, the algorithm terminates (i.e., all nodes return a value) as long as every value in $A$ has non zero probability to be chosen, and the collective output $\val=\{\val(v),v\in V\}$ is such that $(G,\val)\in \cL$ for any graph $G=(V,E)$. The performances of the algorithm (i.e., how many rounds it takes until all nodes terminate) however depends on the choice of the probability distributions $\cD$ chosen by each node $v\in V$ at each round. For instance, in the case of MIS, the algorithm converges if every node~$v$  proposes itself to enter the MIS with probability $p_v=\frac{1}{2}$. However, the choice $p_v = \frac{1}{2\deg(v)}$ as in Luby's algorithm guaranties a much faster convergence: $O(\log n)$ rounds, w.h.p., in $n$-node networks. 

%---------------------------------------------------------------------------
\subsection{A distributed game for LCL languages}
%---------------------------------------------------------------------------

Let $\cL$ be a greedily constructible LCL language of radius $t$. $\cL$ is associated to a game, that we call LCL game, and that we define  as follows (for a more formal definition, see Section~\ref{subsec:formaldefLCLgame}). Let $G$ be a connected simple graph. Every node $v$ of $G$ is a rational and selfish \emph{player}, playing the game with the ultimate goal of maximizing its payoff. Roughly, a \emph{strategy} for a node $v$ is a probability distribution $\cD$ over the values in $A$ compatible with the $t$-ball centered at~$v$ such as observed by $v$, employed at Step~\ref{ins:rand} in Algorithm~\ref{alg:typical-form}, which may depend on the history of~$v$ during the execution of the algorithm. 

As mention in the introduction, we want the payoff of the nodes to reflect the tension between the objective of every node to compute a global solution rapidly (for this global solution brings benefits to every node), and its objective of avoiding certain forms of solutions (which may not be desirable from an individual perspective). So, to define the payoff of a node, we denote by $\pref$ a \emph{preference function}, which is an abstraction of how much a node will ''suffer'' in the future according to a computed solution. For instance, in the MIS game, for a computed MIS $I$, we could have, e.g., $\pref(I)=1$ from a node not in $I$, and $\pref(I)=0$ for a node in~$I$.  In general, for $G=(V,E)$, we define
$
\pref: \good(\cL) \to \mathbb{R}^+
$
by associating to each ball $B\in\good(\cL)$,  the preference $\pref(B)$ of the center of ball $B$. The \emph{payoff function} $\pi_v$ of node $v$ at the completion of the algorithm is decaying with the number $k$ of iterations of the while-loop in Algorithm~\ref{alg:typical-form} before the algorithm terminates. More precisely, we set
$
\pi_v= \delta^{k}\;\pref(B_v)
$ 
where $0< \delta< 1$, $B_v$ is the $t$-ball centered at $v$ returned by Algorithm~\ref{alg:typical-form}, and $k$ is the number of while-loop iterations performed by Algorithm~\ref{alg:typical-form} before all nodes in $B_v$ output their values. The choice of $\delta\in(0,1)$ reflects the tradeoff between the quality of the solution from the nodes' perspective, and their desire to construct a global solution. Note that $k$ is at least the time it took for $v$ to decide a value in $A$, and at most the time it took for the algorithm to terminate at all nodes. If the algorithm does not terminate around $v$, that is, if the $t$-ball around $v$ remains perpetually undecided in at least one of its nodes, then we set  $\pi_v=0$. 

The payoff of a node $v$ will thus be large if all nodes in its $t$-ball decide quickly (i.e., $k$ is small), and if the computed $t$-ball $B_v$ suits node~$v$ (i.e., $\pref(B_v)$ is large). Conversely, if $v$ or another node in its $t$-ball takes many rounds before deciding a value (i.e., $k$ is large), or if node $v$ is not satisfied by the computed solution $B_v$ (i.e., $\pref(B_v)$ is small), then the payoff of $v$ will be small. In particular, if  $\pref(B)=\pref(B')$ for every good balls~$B$ and $B'$, then maximizing the payoff is equivalent to completing the task quickly. Instead, if $\pref(B)$ is very small for some good balls $B$, then nodes might be willing to slow down the completion of the task, with the objective of avoiding being the center of these balls, in order to maximize their payoff. That is, such nodes may bias their distribution $\cD$ towards preferred good balls, even if this is to the price of increasing the probability of conflicting with the choices of neighboring nodes, resulting in more iterations for reaching convergence. 

As mentioned before, a \emph{strategy} of a node at one iteration of Algorithm~\ref{alg:typical-form} is a distribution of probability over the values in $A$, which may depend on the history of the node at any point in time during the execution of Algorithm~\ref{alg:typical-form}. Note that, at the $k$th iteration of the while-loop, every node $v$ has acquired knowledge about its environment at distance $kt$ in the actual graph, where $t$ is the radius of the considered language. Importantly, the strategies depend only on the knowledge accumulated by the nodes along the execution of the algorithm. In particular, the fact that nodes might be given  distinct identities, as in the {\sc local} model, is ignored by the nodes to set their strategies, but used only for distinguishing the nodes in the network, in a way similar to the way classical randomized distributed algorithms perform. In fact, at the beginning of the algorithm, player $v$ does not even know which node she will play in the network, and just knows that the network belong to some given graph family $\cF$. 

A first question of interest is whether there a way to assign a strategy profile to the nodes so that no nodes will have incentive to deviate from their given individual strategies? In other words, do there exist Nash equilibria for LCL games? 

%%%%%%%%%%%%%%%%%%%%%%%%%%%%%%%%%%%%%%%%%%%%%
\section{The nature of LCL games}
\label{sec:rw}
%%%%%%%%%%%%%%%%%%%%%%%%%%%%%%%%%%%%%%%%%%%%%

%---------------------------------------------------------------------------
\subsection{LCL games as forms of extensive games with imperfect information}
%---------------------------------------------------------------------------

Table~\ref{table1} surveys the existence theorems for equilibria in various
game settings, from the  finite strategic games to the extensive games with imperfect
information. We restrict our attention to games with a finite number of players 
(for games with infinitely many players, see~\cite{Pel69}).
Trembling hand equilibria are refinements of sequential equilibria, which are themselves refinements of subgame-perfect equilibria, all of them being Nash equilibria. 

\begin{table}[htb]
  \centering
  {\footnotesize
    \begin{tabulary}{\textwidth}{L|L|L|L}
      & \textbf{Strategic Games} & \textbf{Extensive games with} & \textbf{Extensive games with}\\
      & & \textbf{perfect information} & \textbf{imperfect information}\\
      \hline
      & & Selten, 1965 \cite{Sel65} & Selten, 1975 \cite{Sel75}\\
      \textbf{Finite games} & Nash, 1950 \cite{Nas50} & Subgame-perfect equilibrium & Trembling-hand perfect eq.\\
      & Nash equilibrium & Pure strategies & Behavior strategies\\
      \cline{1-1} 
      \cline{3-4}
      \textbf{Games with} & Mixed strategies & Fudenberg, Levine, 83 \cite{FL83} & Fudenberg, Levine, 83 \cite{FL83}\\
      \textbf{a finite} & & Subgame-perfect equilibrium & Sequential equilibrium\\
      \textbf{action set} & & Pure strategies & Behavior strategies\\
      \hline

      \textbf{Games with} & Fan, 1952 \cite{Fan52} & Harris, 1985 \cite{Har85} & Chakrabarti, 1992 \cite{Cha92}\\
      \textbf{an infinite} & Glicksberg, 1952 \cite{Gli52} & Subgame-perfect equilibrium & Nash equilibrium\\
      \textbf{action set} & Nash equilibrium & Pure strategies & Behavior strategies\\
      & Mixed strategies & & \\    
    \end{tabulary}
  }
  \caption{\sl A summary of main results about the existence of equilibria}
  \label{table1}
\end{table}

In Table~\ref{table1}, we distinguish strategic games (i.e., 1-step games like, e.g., prisoner's dilemma) from extensive games (i.e., game trees with payoffs, like, e.g., monetary policy in economy). For the latter class, we also distinguish games with perfect information (i.e.,  every player knows exactly what has taken place earlier in the game), from the games with imperfect information. LCL games belongs to the class of extensive games with imperfect information, since a node plays arbitrarily many times, and is not necessarily aware of the actions taken by far away nodes in the network. 

In Table~\ref{table1}, we also distinguish finite games (i.e., games with a finite number of pure strategies and finite number of repetitions) from games with infinite horizon (i.e., games which can be repeated infinitely often). The latter class of games is also split into games with finite numbers of actions, and games with infinite set of actions (like, e.g., when fixing the price of a product). LCL games belongs to the class of games with infinite horizon and finite action set: the horizon is infinite since neighboring nodes may perpetually prevent each other from terminating, and the action set $A$ is supposed to be finite.

The previous work most closely related to our framework is~\cite{FL83}. Indeed,  as was said before, LCL games are extensive games with imperfect information, and finite action set. 
 Fudenberg and Levine~\cite{FL83} proved that every such game, under specific assumptions, 
 has a sequential equilibrium (in behavioral strategies). The class of games for which this holds is 
  best described as extensive games with observable actions, simultaneous moves,
  perfect recall and a finite action set, plus some continuity
  requirements. As a consequence, the result in~\cite{FL83} does \emph{not} directly apply to LCL games. 
  Indeed, first, in LCL games the actions of far-away nodes are not observable. Second, the imperfect 
  information in~\cite{FL83} is solely related to the fact that players play simultaneously, while, again, in LCL games, 
  imperfect information also refers to the fact that each  node is not aware of the states of far away nodes in the network. 
Before showing how to extend~\cite{FL83} to LCL games, we discuss next  of the types of equilibria which should be 	sought  in the context of LCL games. 

%---------------------------------------------------------------------------
\subsection{Nash equilibria do not capture efficient distributed computation}
\label{sec:nash}
%---------------------------------------------------------------------------

We consider two important notions of Nash equilibria: the weakest form of equilibrium is the eponym Nash equilibrium while the strongest form of Nash equilibrium is the trembling-hand perfect equilibrium. The following result indicates that the basic notion of Nash equilibria is too weak as far as the design of \emph{efficient} algorithms for LCL tasks is concerned. 

\begin{proposition}\label{theo:nashisbad}
For every $k\geq 1$, there is an LCL game $\Gamma^{(k)}$ for a task solvable in at most $r$ rounds with probability at least $1-\frac{1}{2^r}$ for every $r\geq 1$, and a Nash equilibria for $\Gamma^{(k)}$ which systematically requires at least $k$ rounds to converge. 
\end{proposition}

\begin{proof}
We show that this negative result holds even for a task on the 2-node graph $K_2$.
Recall that, in a strategic game, a strategy $s_i$ of player $i$ is
\emph{strongly dominated} by another strategy $t_i$ if, for every strategy
profile $s'_{-i}$ of the other players, the respective payoffs satisfy 
$
\pi_i(s_i,s'_{-i}) < \pi_i(t_i,s'_{-i}).
$
The strategy $s_i$ is  \emph{weakly dominated} by $t_i$ if
$\pi_i(s_i,s'_{-i}) \leq \pi_i(t_i,s'_{-i})$ for every strategy profile
$s'_{-i}$ and $\pi_i(s_i,s''_{-i}) < \pi_i(t_i,s''_{-i})$ for at least
one strategy profile $s''_{-i}$. Dominated strategies are undesirable, and it
is usually assumed that rational players never play dominated
strategies. Neverthless, while Nash equilibria cannot contain strongly
dominated strategies, they may contain weakly dominated
strategies. The strategic game defined in Fig.~\ref{game1}(left) is an example of such an issue. 
We call this game the \emph{constrained coloring} game. Each of the two players has choice between three 
frequencies, namely green ($G$), red ($R$), and blue $(B)$. To maximize their payoffs, the players should better have one playing blue, and the other playing green. 

\begin{figure}[htb]
  \centering
  \includegraphics[width=0.2\textwidth]{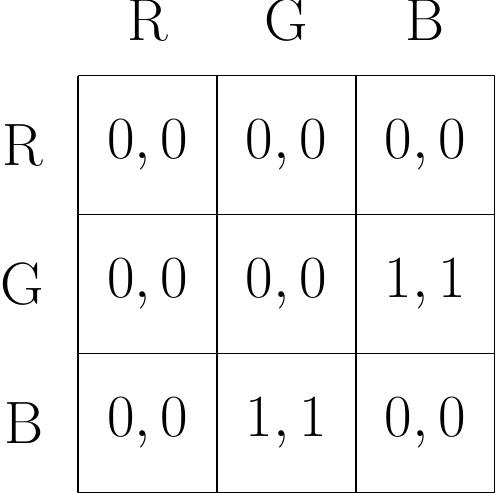}
  \hspace{2cm} \includegraphics[width=0.2\textwidth]{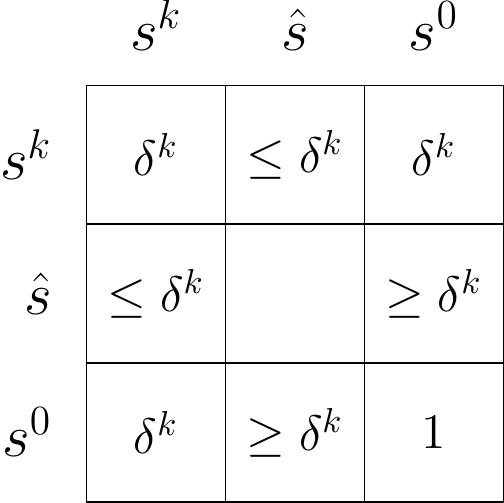}
    %\vspace*{-.3cm}
 \caption{\sl Left: The 2-player strategic game called \emph{constrained coloring}; Right: Payoffs in the LCL game induced by constrained coloring}
  \label{game1}
\end{figure}

The game defined in Fig.~\ref{game1}(left) has four Nash equilibria:
%\begin{itemize}
%\item 
$(R,R)$ with payoffs $(0,0)$,
%\item 
$(G,B)$ with payoffs $(1,1)$,
%\item
 $(B,G)$ with payoffs $(1,1)$,
%\item 
and $(\frac{G + B}{2},\frac{G + B}{2})$ with payoffs $(\frac{1}{2},\frac{1}{2})$.
%\end{itemize}
The Nash equilibrium $(R,R)$ is not realistic since the strategy $R$ is weakly dominated by both $G$ and $B$.
We show that such weakly dominated strategy are generally undesirable in LCL games since they can have dramatic consequences on the convergence time of the algorithm. 

Consider the LCL game played by two players linked by an
edge, inspired from the strategic game  constrained coloring of Fig.~\ref{game1}(left). 
Every node has the choice between three colors: green, blue, and red. A radius-1 ball is good if and only if the two nodes have different colors. However, the cost of the good balls differ. A good ball including color red has low preference. All other good balls have higher preferences.  Specifically, the two players have the same payoff
function depending on their colors. At the first round, the payoffs are

\vspace{-2ex}

\begin{itemize}
\item $(2-\delta,2-\delta)$ when one player is green and the other is blue,

\vspace{-2ex}

\item $(\delta^k,\delta^k)$ when one player is red and the other green or blue,
\end{itemize}

\vspace{-2ex}

\noindent where $\delta \in (0,1)$ is the discount factor of the game, and $k$ is a
positive integer parameter. In the case of a collision, another round of the
game is played, and the payoffs are discounted by a factor~$\delta$. 

The tree structure of the game is represented in Fig.~\ref{game_tree}, where the triangles are 
self-similar representation of the same tree, with payoff discounted by a multiplicative factor $\delta$. 
The tree is thus infinite. The figure also describes the different payoffs
after the first round. 

\begin{figure}[htb]
  \centering
  \includegraphics[width=0.7\textwidth]{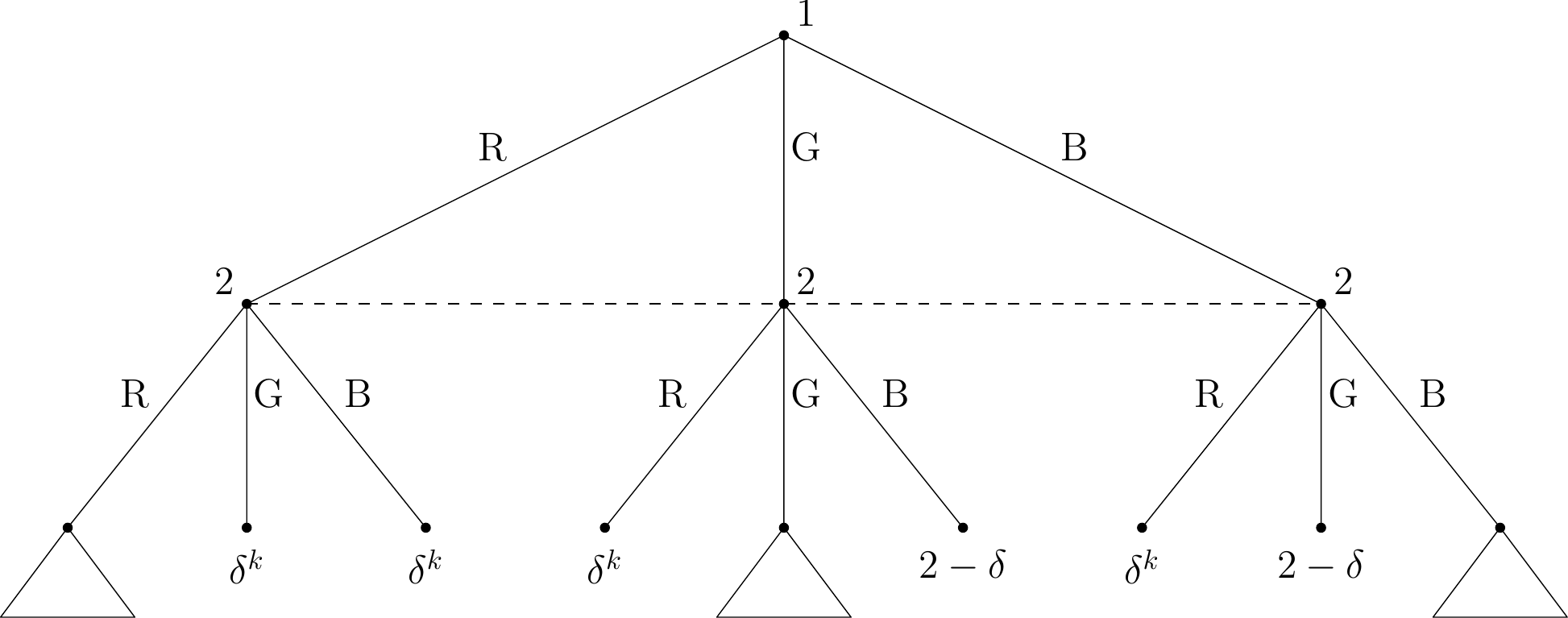}
  \caption{\sl Tree structure of the LCL game inspired by constrained coloring}
  \label{game_tree}
\end{figure}

As it can be read on Fig.~\ref{game_tree}, if the players choose their actions among $(R,G)$, $(R,B)$, $(G,R)$ and $(B,R)$,
 then the game ends after the first round, and each player receives a payoff of $\delta^k$. If the players choose their actions among $(G,B)$ and $(B,G)$, then the game also ends, but   each player receives a payoff of $2 - \delta$.
Finally, if the players choose their actions among $(R,R)$, $(G,G)$ and $(B,B)$, then game
  carries on, with payoffs  discounted by a factor~$\delta$.
  
In this game, since each player is fully informed of the history of the play after any given round,
a pure strategy of a player is  simply an infinite sequence of
actions. Let us denote by $S$ the set of all pure strategies. A mixed strategy of a player is
an infinite sequence of probability distributions over the set of actions
$\{R,G,B\}$. Let us denote by $\hat{S}$ the set of all mixed strategies. From Kuhn's
theorem~\cite{Kuh50}, it follows  that $\hat{S}$ is actually equal to the set of all probability
distributions over $S$. We will pay particular attention to the strategies
\[
s^l = (\underbrace{R,\dotsc, R}_{l}, \frac{G + B}{2}, \frac{G + B}{2}, \dotsc),
\] 
which consist in playing the action $R$ for the first $l$ rounds, and then
playing indefinitely the mixed strategy $\frac{G + B}{2}$. All sequences are
indexed from 0. For every sequence $s$, and for every positive integer $i$,
the sequence $s_{|i}$ denotes the sequence resulting from $s$ after removing the $i$
first actions, and reindexing the actions $(i, i+1, \dotsc)$ in $s_{|i}$ as $(0, 1, \dotsc)$.
Let us now compute the equilibria in the LCL game induced by constrained coloring. We denote by $\Pi: S^2 \mapsto \mathbb{R}$ the expected payoff function of the game (note that the players always receive the same payoff). 

\begin{fact}
  \label{claim:ex1} Let $s \in S$ be a pure strategy. If  $R$ is never played in $s$, then $\Pi(s^0,s) = \Pi(s,s^0) = 1$. Otherwise,  let $t$ be the
  first round in which the action $R$ is played in~$s$. We have 
   \[ \Pi(s^0,s) = \Pi(s,s^0) = 1
  - \left(\frac{\delta}{2}\right)^{t} (1 - \delta^k). \]
\end{fact}

\noindent Indeed, if the first action in $s$ is $R$, then the game ends in the
  first round, and $\Pi(s,s^0) = \delta^k$.  So, let us assume that the first action in $s$ is
  $G$ or $B$. Then, with probability $\frac{1}{2}$, the game ends, and each player
  gets payoff $2 - \delta$. Similarly, with probability $\frac{1}{2}$, another round is
  played. Since $s^0_{|1} =
  s^0$, we get the following recurrence relation: 
  \[ 
  \Pi(s,s^0)
  = \frac{1}{2}\left(2 - \delta\right)
  + \frac{1}{2}\ \delta\ \Pi(s_{|1},s^0). 
  \] 
  This recurrence yields
  \[
    \Pi(s,s^0) = \frac{2 - \delta}{2} \sum_{i = 0}^{t - 1} \left(\frac{\delta}{2}\right)^i + \left(\frac{\delta}{2}\right)^{t} \delta^k
    = \frac{2 - \delta}{2} \left(\frac{1 - (\frac{\delta}{2})^{t}}{1-\frac{\delta}{2}}\right) + \left(\frac{\delta}{2}\right)^{t} \delta^k
  = 1 - \left(1 - \delta^k\right) \left(\frac{\delta}{2}\right)^{t}
  \]
as desired. 

\begin{fact}
  \label{claim:ex2} Let $s \in S$ be a pure strategy. Let $t'$ be
  the first round at which action $G$ or $B$ is played in $s$, and let 
  $t''$ be the first round at which action $R$ is played in
  $s_{|k}$. 
    \[
    \Pi(s^k,s) = \Pi(s,s^k) = \left \{ \begin{array}{ll}
   \delta^{t' + k} & \mbox{if $t' < k $}\\
   \delta^{2 k} & \mbox{if $t' > k $}\\
   \delta^k \left(1 - \left(\frac{\delta}{2}\right)^{t''} (1 - \delta^k)\right) & \mbox{if $t' = k $}
  \end{array} \right .
  \]
\end{fact}

\noindent Indeed, if the first action in $s$ is $G$ or $B$, then the game ends
  in the first round and $\Pi(s,s^k) = \delta^k$. Let us assume that the first action in $s$
  is $R$, then the next round is played, and $\Pi(s,s^k)
  = \delta(s_{|1},s^k_{|1})$. Fact~\ref{claim:ex2} then follows from Fact~\ref{claim:ex1}. 

\medskip

Note that it was sufficient to establish the above facts for pure strategies since, 
for any two mixed strategies $\hat{s}_1$ and $\hat{s}_2$, there exist two pure strategies $s$ and $s'$ in the support of $\hat{s}_1$ such that $\Pi(s,\hat{s}_2) \leq \Pi(\hat{s}_1,\hat{s}_2) \leq \Pi(s',\hat{s}_2)$. The following fact follows from Fact~\ref{claim:ex2}, because $s^k$ is a best response to itself as it yields the highest  payoff $\delta^k$.

\begin{fact}
  The strategy profile $(s^k, s^k)$ is a Nash equilibrium.
\end{fact}

Finally, the following completes our analysis of the LCL game induced by constrained coloring. 

\begin{fact}\label{fact:last}
  The strategy $s^k$ is (weakly) dominated by the strategy $s^0$.
\end{fact}

\noindent To see why the fact holds,  let $s \in S$ be a pure strategy and $t$ be the first round in which the
  action $R$ is played in~$s$. Since $\delta<1$, we get 
$
  \delta^k \leq 1 - \left(1 - \delta^k\right) \left(\frac{\delta}{2}\right)^{t},
$
from which if follows from Fact~\ref{claim:ex1} that $\Pi(s^0,s) \geq \delta^k$.
Furthermore, from Fact~\ref{claim:ex2}, we get that
  $\Pi(s^k,s) \leq \delta^k$. By extension to mixed strategies, we have that, for
  every mixed strategy $\hat{s} \in \hat{S}$, we have 
  $\Pi(s^0,\hat{s}) \geq \Pi(s^k,\hat{s})$.
  It remains to show that there exists a mixed strategy $\hat{s} \in \hat{S}$
  such that $\Pi(s^0,\hat{s}) > \Pi(s^k,\hat{s})$. This is achieved for example
  by $\hat{s} = s^0$. Indeed by Fact~\ref{claim:ex1}, we have
  $
  \Pi(s^0,s^0) = 1 \; \mbox{and} \; \Pi(s^k,s^0) = \delta^k
  $
  which completes the proof of Fact~\ref{fact:last}. 

The above claims are summarized in Fig.~\ref{game1}(right), where $\hat{s}$
represents any mixed strategy not in $\{s^0,s^k\}$. 
The strategy profile $(s^k,s^k)$ is a Nash equilibrium. For this strategy profile,  the LCL game induced by constrained coloring  requires at least $k$ rounds to converge because both players systematically play action $R$ for the first $k$  rounds. On the other hand, The strategy profile $(s^0,s^0)$ allows convergence in constant expected number of rounds, and converges in at most $r$ rounds with probability at least $1-\frac{1}{2^r}$. 
\end{proof}

%---------------------------------------------------------------------------
\subsection{Trembling-hand perfect equilibria do capture efficient distributed computation}
\label{sec:tbheq}
%---------------------------------------------------------------------------

We have seen in Proposition~\ref{theo:nashisbad} that focussing solely on Nash equilibria in a broad sense  is not sufficient. Indeed, such equilibria can include weakly dominated strategies, and this may have a dramatic impact on the convergence time of the algorithm. Such problem can be eliminated by considering the stronger notion of  trembling-hand perfect equilibrium (a subset of Nash equilibria). Indeed, recall that a trembling-hand perfect equilibrium of a game $\Gamma$ is a Nash equilibrium which is the limit of Nash equilibria of perturbed games $\Gamma_\epsilon$, $\epsilon>0$, where $\Gamma_\epsilon$ is essentially the game $\Gamma$ but in which the players can set their strategies only up to some error $\epsilon$. It is known that trembling-hand perfect equilibria avoid pathological scenarios in general extensive games~\cite{FL83,Sel75}, mostly because these equilibria never include weakly dominated strategies, as a weakly dominated strategy can never be a best response to a completely mixed strategy profile. In other words, in the perturbed game $\Gamma_\epsilon$, all players must give a probability at least $\epsilon$ to every action. For instance, in the perturbed LCL game for MIS, no nodes can systematically propose not to be in the MIS, and must propose to enter the MIS with probability at least $\epsilon$ at each round. Similarly, in the perturbed LCL game for $(\Delta+1)$-coloring, every color must have a probability at least $\epsilon$ to be proposed. As a third example, one can show that, for the LCL games $\Gamma^{(k)}$, $k\geq 1$, whose existences are established in the proof of Proposition~\ref{theo:nashisbad}, the trembling-hand perfect equilibria yield algorithms converging quickly. More specifically, we have the following: 

\begin{proposition}\label{theo:theisgood}
For all the LCL games $\Gamma^{(k)}$, $k\geq 1$, on the 2-node graph $K_2$ whose existences are established in Proposition~\ref{theo:nashisbad}, the trembling-hand perfect equilibria yield algorithms converging in at most $r$ rounds, with probability at least $1-\frac{1}{2^r}$, for any $r\geq 1$. 
\end{proposition}

\begin{proof}
Fig.~\ref{examplethe} illustrates the notion of trembling-hand perfect equilibria for the LCL game constrained coloring defined in the proof of Proposition~\ref{theo:nashisbad}. In Fig.~\ref{examplethe}, any point in the triangle BRG represents a distribution $D$ of probability in the set $\{B,R,G\}$ of colors. In $\Gamma_\epsilon$, these distributions are bounded to be at least $\epsilon$ away from the boundaries, hence they are inside the triangle with dotted line borders. In $\Gamma_\epsilon$, all equilibria correspond to strategy profiles where the strategy of every player stands on the bottom line marked in bold in this latter triangle. Hence, the Nash equilibrium  $(s^k,s^k)$ cannot be the limit, when $\epsilon$ goes to~0, of equilibria in $\Gamma_\epsilon$.  Indeed, the corner R cannot be the limit of points in the bold line, as this line  actually converges to the line GB. Therefore, the undesirable Nash equilibrium  $(s^k,s^k)$ is not a trembling-hand perfect equilibrium of the LCL game for  constrained coloring. 

\begin{figure}[htb]
  \centering
  \includegraphics[width=0.35\textwidth]{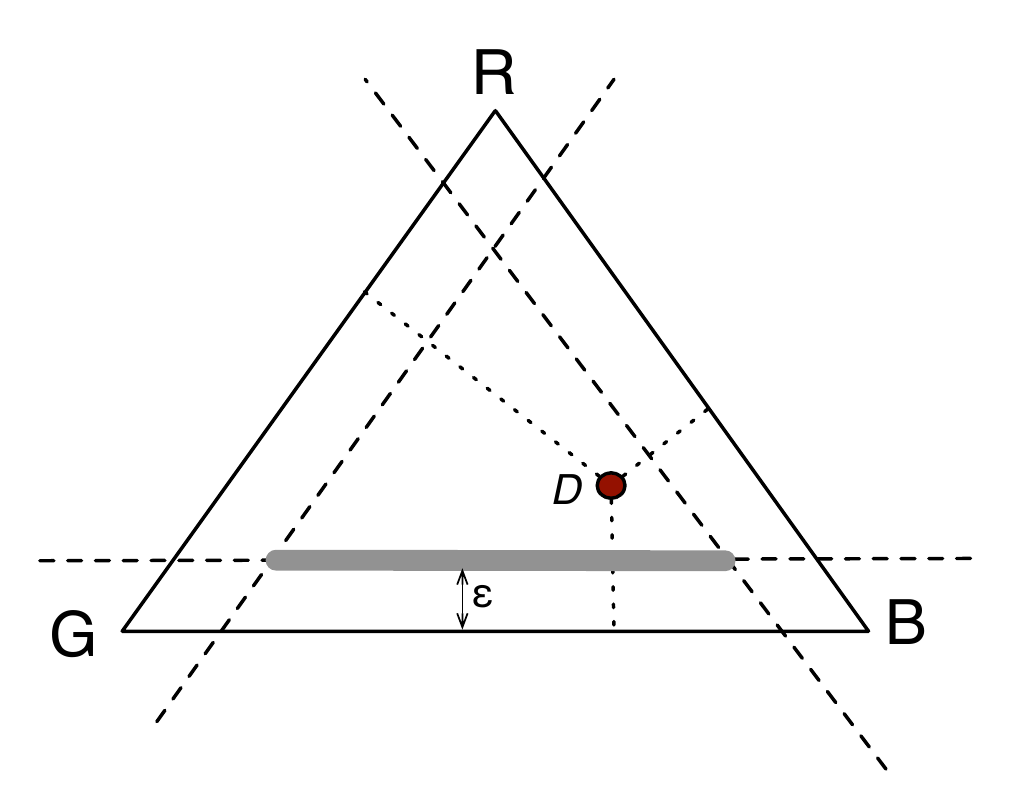}
  \vspace*{-.7cm}
  \caption{Illustration of trembling hand perfect equilibrium}
  \label{examplethe}
\end{figure}

On the other hand, we know from Fact~\ref{fact:last} that $s^k$ is weakly dominated by $s^0$. In particular, the strategy profile $(s^0,s^0)$ is a trembling-hand perfect equilibrium for the game stated in Proposition~\ref{theo:nashisbad}, and this strategy profile allows the game to converge in a constant expected number of rounds. 
\end{proof}

In general, in a perturbed LCL game, no player can block the game for an arbitrarily long time by strictly refusing to select a solution that does not please it. Now, can we always guarantee the existence of trembling hand perfect equilibria in LCL games? In the next sections, we show that the answer to this question is positive. 

%%%%%%%%%%%%%%%%%%%%%%%%%%%%%%%%%%%%%%
\section{Existence of trembling-hand perfect equilibria}
%%%%%%%%%%%%%%%%%%%%%%%%%%%%%%%%%%%%%%

In this section, for the purpose of establishing the existence of  trembling-hand perfect equilibria for LCL-games,  we extend the result of~\cite{FL83} by allowing actions of other players be not observable, which also requires to generalize the notion of rounds as well as the notion of imperfect information, such as used in~\cite{FL83}. 

\medskip

\noindent\textbf{Lemma 1} 
~\textit{Every infinite, continuous, measurable, well-rounded, symmetric extensive game with perfect recall and finite action set has a symmetric trembling-hand perfect equilibrium.}

\medskip

The rest of the section is dedicated to the proof of Lemma~1. We start by formally defining the model. 

%----------------------------------------------------------------
\subsection{The Model}
%----------------------------------------------------------------

We focus our attention on \emph{extensive games} with imperfect information. Our objective is to include \emph{infinite horizon} in the analysis of such games. Recall that an extensive game $\Gamma$ is a tuple $(N, A, X, P, U, p, \pi)$, where:

\begin{itemize}

\item $N = \{1,\dots,n\}$ is the set representing the \emph{players} of the game. An additional player, denoted by $c$, and called \emph{chance}, is modeling all random effects that might occur in the course of the game.

\item $A$ is the (finite) \emph{action set}, i.e., a finite set representing the actions that can be made by each player when she has to play.

\item $X$ is the \emph{game tree}, that is, a subset of $A^* \cup A^{\omega}$ where $A^*$ (resp., $A^{\omega}$) denotes the set of finite (resp., infinite) strings with elements in $A$, satisfying the following properties:

\begin{itemize}
\item the empty sequence $\varnothing\in X$;
\item $X$ is stable by prefix;
\item if $(a_i)_{i= 1,\dots,k}\in X$ for every $k\geq 1$, then $(a_i)_{i\geq 1}\in X$.
\end{itemize}

The set $X$ is partially ordered by the prefix relation, denoted by
$\preceq$, where $x \preceq y$ means that $x$ is a prefix of $y$, and
$x \prec y$ means that $x$ is a prefix of $y$ distinct from $y$. The
elements of $X$ are called \emph{histories}. A history $x$ is
\emph{terminal} if it is a prefix of no other histories in $X$. In
particular, every infinite history in $X$ is terminal. The set of
terminal histories is denoted by $Z$. If the longest history is finite
then the game has \emph{finite horizon}, otherwise it has \emph{infinite horizon}.
For every non-terminal history $x$, we denote by
$$A(x) = \{a \in A : (x,a) \in X\}$$ the set of available actions
after the history $x$.

\item $P$ is the \emph{player partition}, i.e., a function
  $P : X \setminus Z \rightarrow N \cup \{c\}$ that assigns a player
  to each non-terminal history. $P(x)$ is the player who has to play
  after the history~$x$. The sets
  $P_i = \{x \in X \setminus Z : P(x) = i\}$, $i\in N\cup \{c\}$,
  called \emph{player sets}, form a partition of $X \setminus Z$.

\item $U$ is the \emph{information partition}, that is, a refinement of the
  player partition whose elements are called \emph{information sets} such that
  for every $u \in U$, and for every two histories $x, y$ in this
  \emph{information set}~$u$, we have $A(x) = A(y)$, i.e., the sets of available
  actions after $x$ and $y$ are identical.  We can therefore define $A(u)$ as
  the set of actions available after the information set $u$. Formally,
  \[A(u)=\{a\in A : (x,a) \in X \; \mbox{for every} \; x\in u\}.\] For every
  history $x$, the information set containing $x$ is denoted by $u(x)$. We also
  define $P(u)$ as the player who has to play after the information set $u$ has
  been reached, and for every player $i$, the set
  $U_i = \{u \in U : P(u) = i\}$.  $\{U_i, i \in N \cup \{c\}\}$ forms a
  partition of $U$. Information sets regroup histories that are
  indistinguishable to players. Since the player chance is not expected to
  behave rationally, we will simply put $U_c = \{\{x\}, x \in P_c\}$.

\item $p$ is a function that assigns to every history $x$ in $P_c$ (the player
  set of the chance) a probability distribution over the set $A(x)$ of available
  actions after the history $x$. This \emph{chance function} $p$ is supposed to
  be common knowledge among the players.

\item $\pi$ is the \emph{payoff function}, that is,  $\pi : Z \rightarrow \mathbb{R}^n$ assigns the payoff (a real value) to every player in $N$ for every terminal history of the game. We assume that every payoff is in $[-M, M]$ for some $M \geq 0$.

\end{itemize}

\noindent We introduce the concept of \emph{rounds} in extensive games as follows. 

\begin{definition}
  The \emph{round function} $r$ of an extensive game assigns a positive integer
  to every non terminal history $x$, defined by $r(x) = |Rec(x)|$ where
  $ Rec(x) = \{x' \in X\ |\ x' \prec x \text{ and } P(x') = P(x) \}.$ We call
  $r(x)$ the \emph{round} of $x$. The round of a finite terminal history is the
  round of its predecessor, and the round of an infinite history $z$ is
  $r(z) = \infty$. An extensive game $\Gamma$ for which the round function is
  non decreasing with respect to the prefix relation, i.e
  $y \preceq x \Rightarrow r(y) \leq r(x)$, is said to be \emph{well-rounded}.
\end{definition}

Note that not every game is well-rounded, because two histories $x$ and $y$ such
that $x \preceq y$ do not necessarily verify $P(x) = P(y)$. In a well-rounded
game, since $r$ is non decreasing, we have that for any non terminal history
$x$, every player has played at most $r(x) + 1$ times before $x$. Moreover, every player
which has played less than $r(x)$ times before $x$ will never play again after $x$.

Let $u \in U_i$ and $u' \in U_i$ be two (non necessarily distinct) information sets of the same player $i$, for which there
exist $x \in u$, $x' \in u'$, and $a \in A(u')$, such that $(x',a) \preceq x$. 
Recall that an extensive game is said to have \emph{perfect recall} if, for every such $i$,
$u$, $u'$ and $a$, we have: 
\[
\forall y \in u, \exists y' \in u' \mid (y', a)\preceq y.
\]
The following lemma will allow us to safely talk about the round of an information set.

\begin{lemma}
  \label{round_coherence}
  Let $\Gamma$ be an extensive game with perfect recall, and let $x\in X$ and
  $x' \in X$ be two non terminal histories in the same information set~$u\in
  U$. Then $x$ and $x'$ have the same round.
\end{lemma}

\begin{proof}
We first observe the following. Let $\Gamma$ be an extensive game with perfect recall, and let $y\in X$ be a finite history. Let $y'\in X$ and $y''\in X$ for which there is $u\in U$ such that $y'\in Rec(y)\cap u$ and $y''\in Rec(y)\cap u$.  Then $y'=y''$. 
Indeed, since both $y'$ and $y''$ are in $Rec(y)$, we have that both are prefixes of $y$, and thus one of the two is a prefix of the other. Assume, w.l.o.g., that 
  $y'' \prec y' \prec y$ (as, if $y'=y''$ then we are done). Let $a$
  be the action such that $(y'', a) \preceq y'$. Since the game has perfect
  recall, there must exist a history $y'''\in u$ such that
  $(y''', a) \preceq y''$. Thus $y'''\prec y'' \prec y' \prec y$. We can repeat the same reasoning for $y'''$ and $y''$ as we did for $y''$ and $y'$, so on and so forth. In this way, we construct an
  infinite strictly decreasing sequence of histories, which
  contradicts the fact that $y$ is finite.
  
  If both $Rec(x)$ and $Rec(x')$ are empty, then $x$ and $x'$ have the same round. 
  Assume, w.l.o.g., that  $Rec(x)\neq \varnothing$, and let $y \in Rec(x)$. Let $a$ be the action such
  that $(y, a) \preceq x$. Since the game has perfect recall, there exists
  $y' \in u(x')$ such that $(y',a) \preceq x'$. Therefore $y' \prec x'$ and
  $P(y') = P(y) = P(x) = P(x')$. It follows that  $y' \in Rec(x')$. Thus, for any $y\in Rec(x)$, we have found an $y'\in Rec(x')$. This mapping from $Rec(x)$ to $Rec(x')$ is one to one. Indeed, let $y_1$ and $y_2$ in $Rec(x)$, and let $y'_1$ and $y'_2$ in $Rec(x')$ be the corresponding histories. If $y'_1=y'_2$, then, since $u(y_1)=u(y'_1)$ and $u(y_2)=u(y'_2)=u(y'_1)$, we get that $y_1\in Rec(x)\cap u(y'_1)$, $y_2\in Rec(x)\cap u(y'_1)$. It follows from 
 the above observation that $y_1=y_2$. Thus the mapping is one-to-one, and hence $r(x) \leq r(x')$. It follows that we also have $Rec(x')\neq\varnothing$. Therefore, 
  we can apply the same reasoning by switching the roles of $x$ and $x'$, which yields $r(x')\leq r(x)$. Thus
  $r(x) = r(x')$.
\end{proof}

%The following lemma shows that well-rounded games  have a limited number of information sets. 
%
%\begin{lemma}
%  \label{countable}
%  A well-rounded  extensive game with perfect recall has countably many  information sets.
%\end{lemma}
%
%\begin{proof}
%  There are at most $|A|^{n r}$ histories of round~$r$ since, in every history of round~$r$ in a well-rounded game,  every player has played at most $r$ times. 
%  Therefore, there are finitely many information sets of round~$r$. 
%  It follows that one can enumerate the information sets by increasing round.
%\end{proof}

%In the scope of this paper, an extensive game will always refer to an extensive
%game with perfect recall, well-rounded and with a finite action set.

%----------------------------------------------------------------
\subsection{Strategies, Outcomes and Expected Payoff}
\label{subsection:outcomes}
%----------------------------------------------------------------

In this section, we first recall several basic concepts about extensive games with perfect recall. Without loss of generality, we restrict our attention to \emph{behavioral strategies} since such strategies 
are outcome-equivalent to mixed strategies thanks to Kuhn's theorem~\cite{Kuh53}. The main objective of this section is to define the 
  \emph{expected payoff function}, which is novel for it is adapted to the infinite games. 

Recall that the \emph{local strategy} $b_{i,u}$ of a player $i$ given the information
sets $u$ is a probability distribution over the set $A(u)$ of actions available
given this information set. The set of local strategies of player $i$ for $u$ is denoted by $B_{i,u}$.
The  \emph{behavioral strategy} $b_i$ of a player $i$ is a function which assigns a
local strategy $b_{i,u}$ to every information set $u$ of this player. The set of
all behavioral strategies of player $i$ is denoted by $B_i$.
A \emph{strategy profile} is a $n$-tuple of behavioral strategies, one for each
player. The set of all strategy profiles is
$B = \times_{i\in N} B_i$. For each player~$i$,  we denote by $B_{-i}$ the set $\times_{j \neq i}B_j$. Since
$B = B_i \times B_{-i} = \times_{i \in N} B_i$,  a
strategy profile $b$ can be identified in these different ways as 
$b = (b_i, b_{-i}) = (b_1, b_2, \dots, b_n)$.
When every player plays according to a strategy profile $b$, the outcome of the
game in entirely determined, in the sense that every history $x$ has a
probability $\rho_b(x)$ of being reached. For any strategy profile $b$ and any
history $x = (a_i)_{i = 1,\dots,k}$ where $k\in \mathbb{N}\cup \infty$, the
\emph{realization probability} of $x$ is defined by $\rho_b(\varnothing) = 1$,
and
  \[\rho_b(x) = \prod_{i=0}^{k-1} b_{P(x_i),u(x_i)}(a_{i+1})\]
  where $x_0 = \varnothing$, and, for every positive $i \leq k$,
  $x_i = (a_1, a_2, \dots, a_i)$. When the player is the chance, we simply
  identify its strategy with the chance function $p$:
  $P(x_i) = c \Rightarrow b_{P(x_i),u(x_i)}(a_{i+1}) = p(x_i, a_{i+1})$
The function $\rho_b : X \rightarrow [0, 1]$ is called the \emph{outcome} of
the game under the strategy profile $b$. An outcome
$\rho : X \rightarrow [0, 1]$ is \emph{feasible} if and only if there exists a strategy
profile $b$ such that $\rho = \rho_b$. The set of feasible outcomes of $\Gamma$
is denoted by $O$.

We are now ready to define the \emph{expected payoff}.  Note that in a game with
infinite horizon, there can be uncountably many terminal histories. Therefore
the definition of the probability measure on $Z$ requires some care.
For any finite history $x$, let $Z_x$ be the set $\{z \in Z\ |\ x \preceq
z\}$. Note that $Z_x$ might be uncountable. Let $\Sigma$ be the $\sigma$-algebra
on $Z$ generated by
 %all singletons $\{z\}, z \in Z$ and
all sets of the form $Z_x$ for some finite history $x$.
For each strategy profile $b$, the measure $\mu_b$ on $\Sigma$ is defined by:
for every set $Z_x$, $\mu_b(Z_x) = \rho_b(x)$.
%  \begin{itemize}
%  \item for every terminal history $z$, $\mu_b(z) = \rho_b(z)$, and 

%  \end{itemize}
This definition ensures that $\mu_b$ is a probability measure because
$\mu_b(Z) = \mu_b(Z_{\varnothing}) = \rho_b(\varnothing) = 1$.

\begin{definition}
  Let $\pi$ be a  payoff function that is measurable on $\Sigma$. The
  \emph{expected payoff function} $\Pi$ assigns a real value $\Pi(b)$
  to every strategy profile $b \in B$, defined by 
  \[\Pi(b) = \int_\Sigma \pi\ d\mu_b~.\]
\end{definition}

\noindent Note that each component of the expected payoff function is bounded by
$M$, where $M$ is the upper bound on every payoff.
A game $\Gamma$ whose payoff function $\pi$ is measurable on $\Sigma$ is said to
be a \emph{measurable game}. In the following, we always assume that the
considered games are measurable.

%----------------------------------------------------------------
\subsection{$\epsilon$-Equilibria and Subgame Perfection}
%----------------------------------------------------------------

We now show how to adapt  the standard notion of $\epsilon$-equilibria (cf., e.g., Radner \cite{Rad80})  to 
  infinite games (again, Nash equilibrium is a special case of $\epsilon$-equilibrium, with $\epsilon = 0$). 
Recall that a strategy profile $b$ is a $\epsilon$-equilibrium if and only if, for every player $i$, and
  every behavior strategy $b_i' \in B_i$ of this player, we have 
  $\Pi_i(b_i',b_{-i}) - \Pi_i(b) \leq \epsilon$.
  Similarly, we recall the notions of \emph{subgames} and \emph{subgame perfect
    equilibria} (see, e.g., \cite{Sel65}).  A subtree $X'$ of $X$ is said to be
  \emph{regular} if no information sets contain both a history in $X'$ and a
  history not in $X'$. To each regular subtree $X'$ is associated a game
  $\Gamma'=(N, A, X', P', U', p', \pi')$, where $P'$, $U'$, $p'$ and $\pi'$ are
  the restrictions of $P$, $U$, $p$ and $\pi$ to $X'$, called a \emph{subgame}.
  The notions of outcomes and expected payoff functions for subgames follow
  naturally. Specifically,

\begin{definition}
  A strategy profile $b$ is a \emph{subgame perfect $\epsilon$-equilibrium} of
  an infinite game $\Gamma$ if and only if, for every subgame $\Gamma'$, the
  restriction of $b$ to $\Gamma'$ is an $\epsilon$-equilibrium.
\end{definition}

\noindent Note that a subgame perfect $\epsilon$-equilibrium of $\Gamma$ is an
$\epsilon$-equilibrium since $\Gamma$ is a subgame of itself.

%----------------------------------------------------------------
\subsection{Metrics}
\label{subsection:metrics}
%----------------------------------------------------------------

In this section, we now define specific metrics on the set~$O$ of feasible outcomes, and on the set
of behavior strategy profiles. These definitions are inspired from~\cite{FL83}, with adaptations to
  fit our infinite setting. 

\begin{definition}
  \label{outcomes_metric}
  Let $\rho^1, \rho^2 \in O$ be two feasible outcomes of the same extensive
  game $\Gamma$. We define the following metric $d$ on $O$:
  \[d(\rho^1, \rho^2) = \underset{\substack{x \in X\\ x\ \text{finite}}}{\sup}
    2^{-r(x)} \cdot \left| \rho^1(x) - \rho^2(x) \right| \] where $r(x)$ is the
  round of the finite history $x$.
\end{definition}

\begin{lemma}
The function $d:O\times O \to \mathbb{R}$ of Definition~\ref{outcomes_metric} is indeed a metric. 
\end{lemma}

\begin{proof}
  We first show that $d$ satisfies the triangle inequality. Let $\rho^1$,
  $\rho^2$ and $\rho^3$ be three feasible outcomes of $\Gamma$. For any finite
  history $x$ we  have the following:
  \begin{eqnarray*}
  &  |\rho^1(x)-\rho^3(x)| &\leq |\rho^1(x)-\rho^2(x)| + |\rho^2(x)-\rho^3(x)| \\
   \Rightarrow &  2^{-r(x)} |\rho^1(x)-\rho^3(x)| &\leq 2^{-r(x)} |\rho^1(x)-\rho^2(x)| + 2^{-r(x)} |\rho^2(x)-\rho^3(x)|\\
    \Rightarrow &    \underset{\substack{x \in X\\ x\ \text{finite}}}{\sup} 2^{-r(x)} |\rho^1(x)-\rho^3(x)| &\leq \underset{\substack{x \in X\\ x\     \text{finite}}}{\sup} \left(2^{-r(x)} |\rho^1(x)-\rho^2(x)| + 2^{-r(x)} |\rho^2(x)-\rho^3(x)|\right)\\
   \Rightarrow &     \underset{\substack{x \in X\\ x\ \text{finite}}}{\sup} 2^{-r(x)} |\rho^1(x)-\rho^3(x)| &\leq \underset{\substack{x \in X\\ x\ \text{finite}}}{\sup} 2^{-r(x)} |\rho^1(x)-\rho^2(x)| + \underset{\substack{x \in X\\ x\ \text{finite}}}{\sup} 2^{-r(x)} |\rho^2(x)-\rho^3(x)|\\
    \Rightarrow &    d(\rho^1, \rho^3) &\leq d(\rho^1, \rho^2) + d(\rho^2, \rho^3)
  \end{eqnarray*}
  Next, we prove that $d$ separates different outcomes. Let $\rho^1$ and $\rho^2$
  be two feasible outcomes such that $d(\rho^1, \rho^2) = 0$. By definition, this
  implies that, for every finite history $x$, we have $\rho^1(x) =
  \rho^2(x)$. Let $b^1$ and $b^2$ be two strategy profiles such that
  $\rho^1 = \rho_{b^1}$ and $\rho^2 = \rho_{b^2}$. Let
  $x = (a_k)_{k \geq 1}$ be an infinite history, and, for $k\geq 1$, let
  $x_k = (a_1, a_2, \dots ,a_k)$ be the corresponding increasing sequence of its prefixes. By
  definition, $\rho_{b^1}(x)$ is the limit for $k\to \infty$ of the sequence $\rho_{b^1}(x_k)$,
  and $\rho_{b^2}(x_k) \rightarrow \rho_{b^2}(x)$ for $k\to \infty$ as well. Since the two
  sequences are equal they have the same limit, and therefore
  $\rho^1(x) = \rho^2(x)$. Since this is true for any infinite history $x$, it follows that $\rho^1 = \rho^2$.
\end{proof}

We can then use $d$ to define a metric on behavioral strategy profiles as
follows. Let $b^1, b^2 \in B$ be two behavioral strategy profiles of the same
game $\Gamma$. We define the metric $d$ on $B$ by:
  \[d(b^1, b^2) = \max \Big \{ 
  d\left(\rho_{b^1}, \rho_{b^2}\right), \;
  \underset{\substack{i \in N\\b_i \in B_i}}{\sup} d\left(\rho_{(b_i, b^1_{-i})}, \rho_{(b_i, b^2_{-i})}\right)
      \Big \}\]
Finally, we define the continuity of the expected payoff function using the sup
norm over~$\mathbb{R}^n$. Specifically, the expected payoff function $\Pi$ is \emph{continuous} if, for every sequence
  of strategy profiles $(b^k)_{k\geq 1}$ and every strategy profile $b$, we have 
  \[d(b^k, b)  \underset{k\to \infty}{\rightarrow} 0 \implies \underset{i \in N}{\sup}\ \left| \Pi_i(b^k) -
    \Pi_i(b) \right| \underset{k\to \infty}{\rightarrow} 0\]
We then say that an extensive game $\Gamma$ is continuous if its expected payoff
function is continuous. 
% Every finite game is continuous, and the games of our interest are continuous
% games with infinite horizon.

%----------------------------------------------------------------
\subsection{Perturbed Games and Perfect Equilibria}
%----------------------------------------------------------------

In this section, we extend the classical notion of trembling-hand perfect
$\epsilon$-equilibria to infinite games, generalizing~\cite{Sel75}.
We start by revisiting the notion of a perturbed game in the context of infinite games. 
A \emph{perturbed game} $\hat{\Gamma}$ is a pair $(\Gamma, \eta)$, where
$\Gamma$ is an extensive game and $\eta$ is a function which assigns a positive
probability $\eta_u(a)$ to \emph{every} action $a$ available after the information set
$u$ such that
  \begin{equation}
    \label{min_prob}
    \forall u \in U, \quad \sum_{a \in A(u)} \eta_u(a) < 1.
  \end{equation}
The perturbed game $\hat{\Gamma}$ is to be interpreted as a restriction of the
game $\Gamma$ in the following sense. The probability $\eta_u(a)$ is the
\emph{minimal probability} that player $P(u)$ can assign to the action $a$ in
his local strategy at information set $u$. Therefore we define the set of
strategies of the perturbed game as follows.
  In the perturbed game $\hat{\Gamma}$, a \emph{local strategy} $\hat{b}_{i,u}$
  of a player $i$ at one of its information sets $u$ is a probability
  distribution over the set $A(u)$ of actions available at this information
  set satisfying:
\[
    \forall (i,u,a)\in N\times U_i \times A(u), \quad \hat{b}_{i,u}(a) \geq \eta_u(a).
\]
The set of local strategies of player $i$ at the information set $u$ in the
perturbed game $\hat{\Gamma}$ is denoted by $\hat{B}_{i,u}$. Note that
$\hat{B}_{i,u}$ is not reduced to a single strategy thanks to
Eq.~\eqref{min_prob}.

From the above, we can define the set $\hat{B}_i$ of behavioral strategies of
player $i$, the set $\hat{B}$ of strategy profiles, and the notion of
$\epsilon$-equilibria of a perturbed game in the same way as for regular
extensive games.

\begin{definition}
  \label{perfect_eq}
  A strategy profile $b^*$ is a \emph{trembling hand perfect} $\epsilon$-equilibrium of
  $\Gamma$ if 
  \begin{enumerate}
  \item there exists a sequence $(\hat{\Gamma}^k) = (\Gamma, \eta^k)$, $k\geq 1$,  of
  perturbations of $\Gamma$ such that $\eta^k \rightarrow 0$ when $k\to \infty$ (where $\eta^k \rightarrow 0$ means that, for every player $i$, every information set
  $u \in U_i$, and every action $a \in A(u)$, $\eta^k_u(a) \rightarrow 0$), and
  \item there exists  a sequence of
  $\epsilon$-equilibria $\hat{b}^k$ of these perturbed games such that
  $\hat{b}^k \rightarrow b^*$ when $k\to \infty$ (where $\hat{b}^k \rightarrow b^*$  in the sense of the metric on behavioral strategy profiles).
\end{enumerate}
   Such a sequence $(\hat{\Gamma}^k)_{k\geq 1}$ of
  perturbed games is called a \emph{test sequence} for $\Gamma$.
\end{definition}

The next lemma shows that the notion of trembling hand perfect $\epsilon$-equilibrium is a refinement of the notion of subgame
  perfect $\epsilon$-equilibrium.

\begin{lemma}
  Every trembling hand perfect $\epsilon$-equilibrium of a continuous game is a subgame
  perfect $\epsilon$-equilibrium.
\end{lemma}

\begin{proof}
We first show that every trembling hand perfect $\epsilon$-equilibrium of a continuous game is an
  $\epsilon$-equilibrium.  Let $b^*$, $\Gamma$, $\hat{\Gamma}^k$ and $\hat{b}^k$ be as in Definition
  \ref{perfect_eq}. Since $\hat{b}^k$ are $\epsilon$-equilibria of the games
  $\hat{\Gamma}^k$, the following holds:
  \[\forall i \in N, \forall k, \forall b_i \in \hat{B}_i^k, \quad \Pi_i(b_i,
    \hat{b}_{-i}^k) - \Pi_i(\hat{b}^k) \leq \epsilon\] Let $B_i^m$ be the
  intersection of all $\hat{B}_i^k$ with $k \geq m$. Then we have:
  \[\forall i \in N, \forall k \geq m, \forall b_i \in B_i^m, \quad \Pi_i(b_i,
    \hat{b}_{-i}^k) - \Pi_i(\hat{b}^k) \leq \epsilon\] Since $\Pi$ is assumed to be 
  continuous, these inequalities remain valid if on both sides, we take the
  limits for $k \rightarrow \infty$:
  \[\forall i \in N, \forall b_i \in B_i^m, \quad \Pi_i(b_i,
    b_{-i}^*) - \Pi_i(b^*) \leq \epsilon\] This holds for every $m$. The closure
  of the union of all $B_i^m$ is $B_i$. Since $\Pi$ is continuous, this yields:
  \[\forall i \in N, \forall b_i \in B_i, \quad \Pi_i(b_i,
    b_{-i}^*) - \Pi_i(b^*) \leq \epsilon\] as desired. 

Moreover, every $\epsilon$-equilibrium of a perturbed game is subgame perfect. Indeed, 
  this follows from the fact that in a perturbed game, every subgame is reached
  with a strictly positive probability.

  We now have all the ingredients to establish the lemma. 
  Let $b^*$ be a trembling hand perfect $\epsilon$-equilibrium of the game $\Gamma$ and
  $\hat{b}_k$ a sequence of equilibria of the test sequence $\hat{\Gamma}$ such
  that $\hat{b}^k \rightarrow b^*$. Every subgame $\Gamma'$ of $\Gamma$ induces
  a sequence of subgames $\hat{\Gamma}'^k$ of the test sequence. 
  Since every $\epsilon$-equilibrium of a perturbed game is subgame perfect, it follows that 
   the sequence $\hat{b}^k$ induces a sequence of
  equilibria $\hat{b}'^k$ of these pertubed subgames, which converges to ${b'}^*$,
  the strategy profile induced by $b^*$ on $\Gamma'$. Therefore ${b'}^*$ is a
  trembling hand perfect $\epsilon$-equilibrium of $\Gamma'$, which proves that $b^*$ is
  subgame perfect since every trembling hand perfect $\epsilon$-equilibrium of a continuous game is an
  $\epsilon$-equilibrium.
\end{proof}

%----------------------------------------------------------------
\subsection{Truncated Games and Induced Equilibria}
%----------------------------------------------------------------

The goal of this section is to introduce the concept of \emph{truncation} of an
infinite game, and to draw links between equilibria of the truncated (finite)
game and the equilibria of the corresponding infinite game. The purpose of truncation is to
deal with finite game that are easier to handle than infinite ones. The material in this section is a
 generalization of~\cite{FL83} to extensive games (the
  original paper~\cite{FL83}  only applies to repeated games). 

  For every well-rounded extensive game $\Gamma = (N, A, X, P, U, p, \pi)$ with
  perfect recall, and for every positive integer $t$, we define the
  \emph{truncated game}
  $\Gamma_t = (N, A, X_t, \tilde{P}, \tilde{U}, \tilde{p}, \tilde{\Pi})$ as
  follows.
  
  \begin{itemize}
  \item $X_t$ is the set of all histories $x \in X$ of round at most
    $t$. The set of terminal histories $\tilde{Z}$ of $\Gamma_t$ is defined as the set of all
    histories with a predecessor in $X_t$ and no successor in $X_t$.
    
  \item $\tilde{P}$, $\tilde{U}$ and $\tilde{p}$ are  the restrictions of
    $P$, $U$, and $p$ to $X_t$. Note that these restrictions do not break information sets thanks to
    Lemma~\ref{round_coherence}.
    
  \item Strategies and strategy profiles of $\Gamma$ are mapped to strategies and strategy profiles of  $\Gamma_t$ as follows. 
  
    \begin{itemize}
    \item Every strategy profile $b$ of $\Gamma$ is mapped to a strategy profile
      $\tilde{b}$ of $\Gamma_t$ obtained by the restriction of $b$ to the
      information sets in $X_t$. The set of strategy profiles of $\Gamma_t$ is
      denoted by $\tilde{B}$. The strategy profile $\tilde{b}$ is said to be \emph{induced} by $b$.
      
    \item Conversely, every strategy profile $\tilde{b}$ of $\Gamma_t$ is mapped to the strategy
      profile $b$ of $\Gamma$ defined by
      \begin{itemize}
      \item $b_{i,u} = \tilde{b}_{i,u}$ for every information set $u \in \tilde{U}$ and every player $i$,
      \item $b_{i,u}(a) = \frac{1}{|A(u)|}$ for every information set
        $u \notin \tilde{U}$, every player $i$ and every action $a$ (i.e., $b_{i,u}$ is a uniform
      distribution).
      \end{itemize}
      The strategy profile $b$ is said to be induced by $\tilde{b}$.
    \end{itemize}
    
    Note that the first induction mapping is not necessarily one-to-one, while
    the second is not necessarily onto.
    
  \item In order to define the
    new payoff function $\tilde{\pi}$ over the set $\tilde{Z}$, we rather  define it implicitly, via the expected payoff function
    $\tilde{\Pi}$ which is set such that, for any strategy profile $\tilde{b} \in \tilde{B}$ that 
    induces $b \in B$, we have $\tilde{\Pi}(\tilde{b}) = \Pi(b)$.
  \end{itemize}

\begin{lemma}
  \label{equilibrium_characterisation}
  In a continuous well-rounded game with perfect recall $\Gamma$, a strategy
  profile $b^*$ is an $\epsilon$-equilibrium if and only if there exists
  sequences $\epsilon^k$, $t^k$ and $b^k$ such that for every $k$, $b^k$ is an
  $\epsilon^k$-equilibrium in $\Gamma_{t^k}$ and as $k$ goes to infinity,
  $\epsilon^k \rightarrow \epsilon$, $t^k \rightarrow \infty$ and
  ${b'}^k \rightarrow b^*$, where ${b'}^k$ is induced by $b^k$ on $\Gamma$.
\end{lemma}

\begin{proof}
  Given a game $\Gamma$, for every $t\geq 1$, let 
  \[w^t = \max_{i\in N} \underset{\substack{b^1, b^2 \in B \; : \; \tilde{b}^1 =
        \tilde{b}^2}}{\sup} \left|\Pi_i(b^1) - \Pi_i(b^2)\right|\] where
  $\tilde{b}^1 = \tilde{b}^2$ stands for the fact that $b^1$ and $b^2$ induce the
  same strategy profiles over the truncated game $\Gamma_t$.
Note that $w^t \leq 2 M$ for every $t\geq 1$. 

\begin{claim}
  \label{w_t_0}
  A continuous game $\Gamma$ satisfies
  $w^t \underset{t \shortrightarrow \infty}{\longrightarrow} 0$.
\end{claim}

To establish the claim, 
  let $b^1$ and $b^2$ be two strategy profiles of $\Gamma$ which induce the same
  strategy profiles over $\Gamma_t$. This implies that for any history $x$ of
  round at most $t$ and any deviation $b_i$ of some player $i$, the relations
  $\rho_{b^1}(x) = \rho_{b^2}(x)$ and
  $\rho_{(b_i, b^1_{-i})}(x) = \rho_{(b_i, b^2_{-i})}(x)$ holds. Therefore by the
  definition of the metric $d$ on feasible outcomes (cf. Definition~\ref{outcomes_metric}) and its extension to behavioral strategy profiles, we have
  $d(b^1, b^2) \leq 2^{-(t+1)}$. Therefore:
  \[w^t \leq \underset{\substack{i \in n\\ b^1, b^2 \in B\\d(b^1,b^2) \leq 2^{-(t+1)}}}{\sup} \left| \Pi_i(b^1) - \Pi_i(b^2) \right|\]
  And since $\Gamma$ is continuous,
  \[\underset{\substack{i \in n\\ b^1, b^2 \in B\\d(b^1,b^2) \leq 2^{-(t+1)}}}{\sup} \left|
      \Pi_i(b^1) - \Pi_i(b^2) \right| \underset{t \shortrightarrow
      \infty}{\longrightarrow} 0\]
This completes the proof of Claim~\ref{w_t_0}.

\begin{claim}
  \label{induced_equilibria_properties}
  Let $\Gamma$ be a well-rounded game with perfect recall and $\Gamma_t$ its
  truncation. We have that:
  \begin{enumerate}
  \item Any $\epsilon$-equilibrium $\tilde{b}^*$ in $\Gamma_t$ induces a
    $(\epsilon + w^t)$-equilibrium in $\Gamma$.
  \item Any $\epsilon$-equilibrium $b^*$ in $\Gamma$ induces a
    $(\epsilon + 2 w^t)$-equilibrium in $\Gamma_t$.
  \end{enumerate}
\end{claim}

To establish the first item of the claim, let $\tilde{b}^*$ be an $\epsilon$-equilibrium in $\Gamma_t$ wich
    induces $b^*$ in $\Gamma$. Let $i$ be a player, $b'_i \in B_i$ a deviation
    of this player in $\Gamma$, $\tilde{b}_i$ the deviation it induces in
    $\Gamma_t$ and $b_i$ the deviation induced by $\tilde{b}_i$ in
    $\Gamma$. Since $\tilde{b}^*$ is an $\epsilon$-equilibrium in $\Gamma_t$, we
    obtain that
    \[\tilde{\Pi}_i(\tilde{b}_i, \tilde{b}^*_{-i}) - \tilde{\Pi}_i(\tilde{b}^*)
      \leq \epsilon\]
    By definition of $\tilde{\Pi}$, the following holds:
    \begin{align*}
      \tilde{\Pi}_i(\tilde{b}_i, \tilde{b}^*_{-i}) &= \Pi_i(b_i, b^*_{-i})\\
      \tilde{\Pi}_i(\tilde{b}^*) &= \Pi_i(b^*)
    \end{align*}
    Furthermore, by definition of $w^t$, we have 
    \[\Pi_i(b'_i, b^*_{-i}) - \Pi_i(b_i, b^*_{-i}) \leq w^t\]
    We eventually arrive at
    \[\Pi_i(b'_i, b^*_{-i}) - \Pi_i(b^*) \leq \epsilon + w^t\]
    Which proves that $b^*$ is a $(\epsilon + w^t)$-equilibrium, since $i$ is
    any player and $b'_i$ is any deviation in $B_i$.
  
  To establish the second item of Claim~\ref{induced_equilibria_properties}, let $b^*$ be an $\epsilon$-equilibrium in $\Gamma$, $\tilde{b}^*$ the
    strategy profile it induces on $\Gamma_t$, and ${b'}^*$ the strategy profile
    induced by $\tilde{b}^*$ on $\Gamma$. Let $i$ be a player,
    $\tilde{b}_i \in \tilde{B}_i$ a deviation in $\Gamma_t$ and $b_i$ the
    deviation it induces in $\Gamma$.
    Because $b^*$ is an $\epsilon$-equilibrium, we have:
    \[\Pi_i(b_i,b_{-i}^*) - \Pi_i(b^*) \leq \epsilon\]
    By definition of $\tilde{\Pi}$, we have:
    \begin{align*}
      \tilde{\Pi}_i(\tilde{b}^*) &= \Pi_i({b'}^*)\\
      \tilde{\Pi}_i(\tilde{b}_i,\tilde{b}_{-i}^*) &= \Pi_i(b_i, {b'_{-i}}^*) 
    \end{align*}
    And by definition of $w^t$:
    \begin{align*}
      \Pi_i({b'}^*) - \Pi_i(b^*) &\leq w^t\\
      \Pi_i(b_i,{b'}_{-i}^*) - \Pi_i(b_i, b_{-i}^*) &\leq w^t
    \end{align*}
    The above relations put together eventually gives us:
    \[\tilde{\Pi}_i(\tilde{b}_i, \tilde{b}_{-i}^*) - \tilde{\Pi}_i(\tilde{b}^*) \leq \epsilon + 2 w^t\]
    Wich proves that $\tilde{b}^*$ is a $(\epsilon + 2 w^t)$-equilibrium in
    $\Gamma_t$.
This completes the proof of Claim~\ref{induced_equilibria_properties}.

Note that Claim~\ref{induced_equilibria_properties} applies to perturbed games as well, as long as induced
strategy profiles verify the constraints of the minimal probabilities.

\begin{claim}
  \label{closed_equilibria}
  Let $b^k$ be a sequence of $\epsilon$-equilibria of a continuous well-rounded
  game with perfect recall $\Gamma$ such that $b^k \rightarrow b^*$. Then $b^*$
  is also an $\epsilon$-equilibrium of $\Gamma$. In other words, for every
  $\epsilon$, the set of $\epsilon$-equilibria of such a game is closed.
\end{claim}

 To establish the claim, suppose that $b^*$ is not an $\epsilon$-equilibrium of $\Gamma$. Therefore for
  some player $i$ and some $\delta > 0$, there exists a deviation $b_i \in B_i$
  such that:
  \[\Pi_i(b_i,b_{-i}^*) - \Pi_i(b^*) \geq \epsilon + 3 \delta\]
  By continuity of $\Pi$, we have for a large enough $k$:
  \begin{align*}
    \Pi_i(b^k) - \Pi_i(b^*) &\leq \delta\\
    \Pi_i(b_i, b_{-i}^*) - \Pi_i(b_i, b_{-i}^k) &\leq \delta
  \end{align*}
  Substracting the two above inequalities to the previous one yields:
  \[\Pi_i(b_i, b_{-i}^k) - \Pi_i(b^k) \geq \epsilon + \delta\]
  which contradicts the premise that $b^k$ is an $\epsilon$-equilibrium.
This completes the proof of Claim~\ref{closed_equilibria}.

We have now all ingredients to prove the lemma. 
Assume first that $b^*$ is an $\epsilon$-equilibrium in $\Gamma$ and define
    $b^k$ to be the strategy profile induced by $b^*$ on $\Gamma_k$ and ${b'}^k$
    induced by $b^k$ on $\Gamma$. By Claim \ref{induced_equilibria_properties},
    for every $k$, $b^k$ is a $(\epsilon + 2 w^k)$-equilibrium in $\Gamma_k$,
    and by Claim \ref{w_t_0}, $w^k \rightarrow 0$. Furthermore,
    $d({b'}^k, b^*) \leq 2^{-k}$ by definition of $d$, therefore
    ${b'}^k \rightarrow b^*$.
 
 Conversely, assume now that there exists such sequences $\epsilon^k$, $t^k$ and
    $b^k$. By Claim \ref{induced_equilibria_properties}, for every $k$, the
    strategy profile ${b'}^k$ induced by $b^k$ on $\Gamma$ is an
    $(\epsilon^k + w^{t^k})$-equilibrium in $\Gamma$. Since $\Gamma$ is
    continuous, we know by Claim \ref{w_t_0} that
    $(\epsilon^k + w^{t^k}) \rightarrow \epsilon$. For every $\delta > 0$, we
    have for every $k$ large enough
    $\epsilon^k + w^{t^k} \leq \epsilon + \delta$. By Claim
    \ref{closed_equilibria}, this implies that $b^*$ is a
    $(\epsilon + \delta)$-equilibrium. Since this is true for every
    $\delta > 0$, we have shown that $b^*$ is an $\epsilon$-equilibrium in
    $\Gamma$.
\end{proof}

%----------------------------------------------------------------
\subsection{Truncation of Perturbed Games}
%----------------------------------------------------------------

The goal of this section is to extend the results of the previous section to
trembling hand perfect $\epsilon$-equilibria (instead of just for
$\epsilon$-equilibria). We define what is the truncation of a perturbed game in
a straightforward way: for every perturbed game $\hat{\Gamma} = (\Gamma, \eta)$
and every positive integer $t$, the \emph{truncated perturbed game}
$\hat{\Gamma}_t$ is $(\Gamma_t, \tilde{\eta})$, where $\tilde{\eta}$ is simply the
restriction of $\eta$ to $X_t$.

\begin{lemma}
  \label{perfect_characterisation}
  In a continuous game $\Gamma$, a strategy profile $b^*$ is a trembling hand perfect
  $\epsilon$-equilibrium if and only if there exists sequences $\epsilon^k$,
  $t^k$ and $b^k$ such that for every $k$, $b^k$ is a trembling hand perfect
  $\epsilon^k$-equilibrium in $\Gamma_{t^k}$, and, as $k$ goes to infinity,
  $\epsilon^k \rightarrow \epsilon$, $t^k \rightarrow \infty$ and
  ${b'}^k \rightarrow b^*$, where ${b'}^k$ is induced by $b^k$ on $\Gamma$.
\end{lemma}

\begin{proof}
We first show the following. 

\begin{claim}
  \label{induced_perfect_equilibria}
  Let $\Gamma$ be a continuous game and $\Gamma_t$ its truncation.
  \begin{enumerate}
  \item Any trembling hand perfect $\epsilon$-equilibrium $\tilde{b}^*$ in $\Gamma_t$ induces a
    trembling hand perfect $(\epsilon + w^t)$-equilibrium in $\Gamma$.
  \item Any trembling hand perfect $\epsilon$-equilibrium $b^*$ in $\Gamma$ induces a trembling hand perfect
    $(\epsilon + 2 w^t)$-equilibrium in $\Gamma_t$.
  \end{enumerate}
\end{claim}

For proving the first item of Claim~\ref{induced_perfect_equilibria}, let $\tilde{b}^*$ be a trembling hand perfect $\epsilon$-equilibrium in $\Gamma_t$. By
    definition of a trembling hand perfect equilibrium, there exists $\tilde{b}^k$ a sequence
    of $\epsilon$-equilibrium of a test sequence
    $\hat{\Gamma}_t^k = (\Gamma_t, \tilde{\eta}^k)$ of the game $\Gamma_t$, such
    that $\tilde{b}^k \rightarrow \tilde{b}^*$. This test sequence can be
    extended into a sequence of infinite games
    $\hat{\Gamma}^k = (\Gamma, \eta^k)$ by choosing the same minimal
    probabilities for any information set of round at most $t$, and for every
    information set $u$ of round greater than $t$ and any action $a \in A(u)$,
    define $\eta_u^k(a) = \frac{2^{-k}}{|A(u)|}$. Therefore the sequence
    $\hat{\Gamma}^k$ is by definition a test sequence of $\Gamma$. Furthermore,
    $\eta^k$ was defined such that Claim \ref{induced_equilibria_properties}
    applies: the strategy profiles $b^k$ induced by $\tilde{b}^k$ are therefore
    $(\epsilon + w^t)$-equilibria of $\hat{\Gamma}_t^k$. We also have that $b^k$
    is induced by $\tilde{b}^k$, $b^*$ is induced by $\tilde{b}^*$ and
    $\tilde{b}^k \rightarrow \tilde{b}^*$, therefore $b^k \rightarrow b^*$.
    This proves that $b^*$ is a trembling hand perfect $(\epsilon + w^t)$-equilibrium of
    $\Gamma$.
  
  For proving the second item of Claim~\ref{induced_perfect_equilibria}, let $b^*$ be a trembling hand perfect $\epsilon$-equilibrium in $\Gamma$. By
    definition, there exists a sequence $b^k$ of $\epsilon$-equilibria of a test
    sequence $\hat{\Gamma}^k$ of $\Gamma$ such that $b^k \rightarrow b^*$. For
    $k$ large enough, the following is true:
    \[\forall u \in U, \forall a \in A(u), \quad \eta_u^k(a) \leq \frac{1}{|A(u)|}\]
    Therefore we can apply Claim \ref{induced_equilibria_properties}: the
    strategy profiles $\tilde{b}^k$ induced by $b^k$ in $\hat{\Gamma}_t^k$ are
    $(\epsilon + 2 w^t)$-equilibria. Obviously $\hat{\Gamma}_t^k$ is a test
    sequence for $\Gamma_t$. And since $\tilde{b}^k$ is induced by $b^k$,
    $\tilde{b}^*$ is induced by $b^*$ and $b^k \rightarrow b^*$, we know that
    $\tilde{b}^k$ converges to $\tilde{b}^*$. Which proves that $\tilde{b}^*$ is
    a $(\epsilon + 2 w^t)$-equilibrium in $\Gamma_t$.
This completes the proof of Claim~\ref{induced_perfect_equilibria}.

\begin{claim}
  \label{closed_perfectness}
  Let $b^k$ be a sequence of trembling hand perfect $\epsilon$-equilibria of a game $\Gamma$
  such that $b^k \rightarrow b^*$. Then $b^*$ is also a trembling hand perfect
  $\epsilon$-equilibrium of $\Gamma$.
\end{claim}

Indeed, for every $k$, since $b^k$ is a trembling hand perfect $\epsilon$-equilibrium, there exists a
  sequence $b^{k,n}$ of $\epsilon$-equilibria of a test sequence
  $\hat{\Gamma}^n$ of $\Gamma$, such that
  $b^{k,n} \underset{n \shortrightarrow \infty}{\longrightarrow} b^k$. The
  sequence $b^{k,k}$ of $\epsilon$-equilibria of the test sequence
  $\hat{\Gamma}^k$ converges to $b^*$, thus proving that it is a trembling hand perfect
  $\epsilon$-equilibrium of $\Gamma$.

  The proof of the lemma is then identical to the proof of Lemma \ref{equilibrium_characterisation}, by replacing Claims \ref{induced_equilibria_properties} and \ref{closed_equilibria}
  with  Claims \ref{induced_perfect_equilibria} and
  \ref{closed_perfectness}, respectively.
\end{proof}

%----------------------------------------------------------------
\subsection{Proof of Lemma~\ref{theo:maingeneral}}
%----------------------------------------------------------------

We have now all the ingredients to prove that, as stated in Lemma~\ref{theo:maingeneral} every infinite, continuous, measurable, well-rounded, extensive game with perfect
  recall and finite action set has a trembling hand perfect equilibrium.
%  
%\begin{theorem}
%  \label{theorem:existence}
%  Every infinite, continuous, measurable, well-rounded, extensive game with perfect
%  recall and finite action set has a trembling hand perfect equilibrium.
%\end{theorem}
%
%\begin{proof}
  First observe that the set of strategy profiles $B$ of an extensive game
  $\Gamma$ with a finite action set is sequentially compact. Indeed, let
  $\Gamma$ be a game with a finite action set $A$. For every player $i$ and
  every information set $u$, the set of local strategies $B_{i u}$ is a simplex
  in the space $\mathbb{R}^{A(u)}$, and therefore it is sequentially
  compact. Because a countable product of sequentially compact spaces is
  sequentially compact, % \cite{SeqCom}
  $$B = \underset{i \in N}{\times} \underset{u \in U_i}{\times} B_{i u}$$ is
  sequentially compact. Indeed, there are countably many finite histories in any
  game with a finite action set, therefore there are also countably many
  information sets.

  Consider now a (possibly infinite) continuous well-rounded extensive game
  $\Gamma$ with perfect recall and a finite action set, and consider the
  corresponding sequence of truncated games $\Gamma_k$. Selten~\cite{Sel75} has
  shown in 1975 that every finite extensive game with perfect recall has a
  trembling hand perfect equilibrium.  By this result, every game $\Gamma_k$ has
  a trembling hand perfect equilibrium that we call $b^k$. We call ${b'}^k$ the
  sequence of strategy profiles induced by $b^k$ in the game $\Gamma$. Because
  $B$ is sequentially compact, we can extract a convergent subsequence
  ${b'}^{t_k}$ and we call $b^*$ its limit. Lemma \ref{perfect_characterisation}
  shows that $b^*$ is a trembling hand perfect equilibrium of $\Gamma$.
  
  %% Added by Pierre %%%
  We say that $\Gamma$ is a \emph{symmetric game} if it satisfies the
  following requirements. 
  (1) 
   For every player $i$, there exist a bijection $type_i: U_i \rightarrow
    [1..I]$ such that for every two players $i$ and $i'$ and two information
    sets $u \in U_i$ and $u' \in U_{i'}$ such that $type_i(u) = type_{i'}(u')$,
    we also have $A(u) = A(u')$ and $\actions_i(u) = \actions_{i'}(u')$.
(2)
  For every two players $i$ and $i'$, we define the relation $R_{i, i'}$
    on strategy profiles: For every two strategy profiles $b$ and $b'$, we have
    $b\ R_{i, i'}\ b'$ if and only if, for every positive integer~$k$, 
    \[b_{i,
      type_i^{-1}(k)} = b'_{i', type_{i'}^{-1}(k)} \;\mbox{and}\; b'_{i, type_i^{-1}(k)}
    = b_{i', type_{i'}^{-1}(k)},
    \] and, for every player $j$ different from $i$ and
    $i'$, $b_j = b'_j$.
  To be symmetric, $\Gamma$ must satisfy that, for every
    two players $i$ and $i'$ and every two strategy profiles $b$ and $b'$ such
    that $b\ R_{i, i'}\ b'$, we have  $\Pi_i(b) = \Pi_{i'}(b')$ and
    $\Pi_i(b') = \Pi_{i'}(b)$.
  In a symmetric game $\Gamma$, we say that a profile of strategies $b \in B$ is
  \emph{symmetric} if and only if for every two players $i$ and $i'$, $b\ R_{i,
    i'}\ b$.
Since a limit of a sequence of symmetric strategy profiles is
  symmetric, we can derive that $b^*$ is a symmetric trembling-hand perfect
  equilibrium. \qed
%\end{proof}

%%%%%%%%%%%%%%%%%%%%%%%%%%%%%%%%%%%%%%%%%%%%%
\section{Existence of efficient robust algorithms for LCL games}
%%%%%%%%%%%%%%%%%%%%%%%%%%%%%%%%%%%%%%%%%%%%%

The hypotheses regarding the topological nature of the strategy, and on the nature of the payoff function (continuity, measurability, etc.) are standard in the framework of extensive games. The notion of well-rounded game is new, and used to capture the fact that the nodes play in rounds in an LCL game. The fact that the equilibrium is symmetric is crucial since, in LCL games, as in randomized distributed computing in general, the instructions given to all nodes are identical, and the behavior of the nodes only vary along with the execution of the algorithm when they progressively discover their environment. We show that LCL games satisfy all hypotheses of Lemma~\ref{theo:maingeneral}, from which we get our main result as stated in the introduction: 

\medskip

\noindent\textbf{Theorem 1} 
~\textit{Let $\cL$ be a greedily constructible locally checkable labeling. The LCL game associated to~$\cL$ has a symmetric trembling-hand perfect equilibrium. }

\medskip

The rest of the section is dedicated to the proof of the theorem. We start by formally defining LCL games. 

\subsection{Formal Definition of a LCL Game}
 \label{subsec:formaldefLCLgame}
 
Let $A$ be a finite alphabet, $\mathcal{F}$ a family of graphs with at most
$n$ vertices, $\mathbf{D}$ a probability distribution over~$\mathcal{F}$, and
$\mathcal{L}$ a greedily constructible LCL language over $\mathcal{F}$. Let $t$
be the radius of $\mathcal{L}$ and $\good(\mathcal{L})$ be the set of good balls in
$\mathcal{L}$. Let $\pref: \good(\mathcal{L}) \mapsto [0,M]$, for $M>0$, be a function
representing the preferences of the players over good balls and $\delta \in
(0,1)$ a \emph{discounting factor}. We define the game
\[
\Gamma(\mathcal{L},\mathbf{D},\pref,\delta) = (N, A, X, P, U, p, \pi)
\]
associated to the language $\mathcal{L}$, the distribution $\mathbf{D}$,
the preference function $\pref$, and the discounting factor $\delta$, as
follows.

\begin{itemize}
\item The player set is $N = \{1, \dots, n\}$.

\item The action set is $A  \cup \mathcal{F}$ where the actions in
  $\mathcal{F}$ are only used by the extra player Chance in the initial move, and the actions in
  $A$ are used by the actual players.
  
\item The first move of the game is made by Chance (i.e $P(\varnothing) = c$). As a result, a 
  graph $G \in \mathcal{F}$ with a mapping of the players to the nodes of $G$ is selected at random according to the probability
  distribution $\mathbf{D}$. From now on, the players are identified with
  the vertices of the graph $G$, labeled from~1 to~$n$. Note that $\cF$ might be reduced to a single graph, e.g., $\cF=\{C_n\}$, and Chance just selects which vertices of the graph will be played by which players. 
  
\item The game is then divided into rounds (corresponding to the intuitive
  meaning in synchronous distributed algorithms). At each round, the
  \emph{active} players play in increasing order, from~1 to~$n$. At round~0 every player is
  active and plays, and  every action in $A$ is available.
  
\item At the end of each round (i.e., after every active player has played the
  same number of times), some players might become \emph{inactive}, depending on
  the actions chosen during the previous rounds. For every $i \in N$, let $s(i)$
  denote the last action played by player~$i$, which we call the \emph{state} of
  $i$, and let $\ball(i)$ denote the ball of radius $t$ centered at node $i$. Every player $i$ such that 
  $ball(i) \in good(\mathcal{L})$ at the end of
  a round becomes \emph{inactive}.
  
\item In subsequent rounds, the set of available actions might be restricted. 
  For every round $r > 0$, and for every active player $i$, an action $a\in A$ is
  available to player $i$ if and only if there exists a ball $b \in
  \good(\mathcal{L})$ compatible with the states of inactive players in which
  $s(i) = a$.
  
\item A history is terminal if and only if either it is infinite or it comes after the
  end of a round and every player is inactive after that round.
  
\item Let $x$ be a history. 
We denote by
  $\actions_i(x)$ the sequence of actions extracted from $x$ by selecting all
  actions taken by player $i$ during rounds before $r(x)$. (The action
  possibly made by player $i$ at round $r(x)$, and  actions made by a player
  $j \neq i$ are not included in $\actions_i(x)$). 
  
\item Let $x$ and $y$ be two non terminal histories such that $P(x) = P(y) =
  i$. Then $x$ and $y$ are in the same information set if and only if for every $ j \in \ball(i)$,
  we have  
  \[
    \quad \actions_j(x) = \actions_j(y).
  \]
  This can be interpreted by saying that a player $i$ ``knows''  every
  action previously taken by any player at distance at most $t$.  
  
\item Let $i$ be a player, and let $z$ be a terminal history. We define the
  \emph{terminating time} $\time_i(z)$ of player $i$ in history $z$ by
  \[
  \time_i(z) = \max \{|\actions_j(z)|, j \in \ball(i)\} - 1.
  \] 
  The payoff function
  $\pi$ of the game is then defined as follows. For every player $i$, and every terminal
  history $z$,
  \[
    \pi_i(z) = \delta^{\time_i(z)} \cdot \pref(\ball(i))
  \]
  And $\pi_i(z) = 0$ if $\time_i(z) = \infty$.
% Comments on the payoff function ?
\end{itemize}

\subsection{Properties of LCL Games}

\begin{fact}
  LCL games are well-rounded.
\end{fact}

\begin{proof}
  This follows directly from the fact that, in a LCL game, every active player plays at
  every round until it becomes inactive, and, once inactive, a player cannot become active
  again. 
\end{proof}

\begin{fact}
  LCL games have perfect recall.
\end{fact}

\begin{proof}
  Let $\Gamma(\mathcal{L},\mathbf{D},\pref,\delta) = (N, A, X, P, U, p, \pi)$ be
  an LCL game. Let $u$ and $u'$ be two information sets of the same player $i$,
  for which there exists $x \in u$, $x' \in u'$ and $a \in A(u')$, such that
  $(x',a) \preceq x$. Let $y$ be a history in $u$.
  Since $x$ and $y$ are in the same information set $u$, it follows that, for
  every player $j\in \ball(i)$, we have $\actions_j(x) = \actions_j(y)$. In particular, this
  implies that $x$ and $y$ are in the same round. Let $y'$ be the only history
  which is a prefix of $y$ with $P(y') = i$ and $r(y') = r(x')$. (Such a
  history exists because $r(x') < r(x)$ and $r(y) = r(x)$).
  Since the players play in the same order at every round, we get that, 
  for every player $j\in \ball(i)$, $\actions_j(x') = \actions_j(y')$. 
  Therefore $y' \in u'$. Furthermore, since $\actions_i(x) = \actions_i(y)$, the
  action played by $i$ after $y'$ must be $a$, which implies $(y',a) \preceq y$,
  and concludes the proof.
\end{proof}

\begin{fact}
  The payoff function $\pi$ of a LCL game is measurable on the $\sigma$-algebra
  $\Sigma$ corresponding to the game. 
\end{fact}

\begin{proof}
    We prove that, for every player $i$ and every $a \in \mathbb{R}$,
    $\pi_i^{-1}(]a,+\infty[) \in \Sigma$, which implies that $\pi$ is measurable
    on $\Sigma$.
    Without loss of generality, let us assume that $M=1$, that is, $\pi_i: Z \mapsto [0,1]$.  For every
    $a<0$, we have $\pi_i^{-1}(]a,+\infty[) = Z$, and thus is in $\Sigma$. Similarly,  for every $a > M$,
    we have $\pi_i^{-1}(]a,+\infty[) = \varnothing \in \Sigma$.
    So, let us assume that $a \in ]0,1]$ and let $z$ be a terminal history such that
    $\pi_i(z) > a$, which implies that $\time_i(z) < \frac{\ln a}{\ln \delta}$,
    i.e every player in $\ball(i)$ has played only a finite number of times in
    the history $z$. Let $x$ be the longest history such that $x \preceq z$ and
    $r(x) =\time_i(z)$. Then the history $x'$ that comes just after $x$ in $z$
    is the shortest prefix of $z$ such that every player in $\ball(i)$ is now
    inactive.
    Let $z'$ be a terminal history such that $x' \preceq z'$. Since every player
    $j \in \ball(i)$ is inactive after $x'$, it follows that the state of any
    such player in $z'$ is the same as its state in $z$. Therefore $\pi_i(z') =
    \pi_i(z)$.
    It follows from the above that, for any terminal history $z$ such that $\pi_i(z) > a$,
    there exists a finite history $x'$ in round $\time_i(z) + 1$ such that
    $z \in Z_{x'} \subseteq \pi_i^{-1}(]a,+\infty[)$. Since there can only be a finite
        number of histories in round $\time_i(z)$, we get  that
        $\pi_i^{-1}(]a,+\infty[)$ is the union of a finite number of sets of the
        form $Z_{x'}$. As a consequence, it is measurable in $\Sigma$.
        It remains to prove that $\pi_i^{-1}(]0,+\infty[) \in \Sigma$. actually, this 
        simply follows from the fact that $\pi_i^{-1}(]0,+\infty[) = \underset{k
          \geq 1}{\cup} \pi_i^{-1}(]\frac{1}{k},+\infty[)$, and the fact that that
        $\Sigma$ is stable by countable unions.
\end{proof}

\begin{fact}
  LCL games are continuous.
\end{fact}

\begin{proof}
  Let $b$ be a strategy profile, and let $(b^k)_{k \geq 0}$ be a sequence of strategy
  profiles such that
  $d(b^k,b)\rightarrow~0$ when $k \rightarrow \infty$. By
  definition of the metric $d$ on $B$ (cf. subsection \ref{subsection:metrics}),
 we have that
  $d(\rho_{b^k},\rho_b)\rightarrow
  0$ when  $k \rightarrow \infty$. By definition of the metric on $O$, we have that, for any finite
  history $x$,
  \[
  |\rho_{b^k}(x) - \rho_b(x)| \underset{k \shortrightarrow
    \infty}{\longrightarrow} 0.
    \]
     It follows that for any set of the form $Z_x$
  as defined in subsection \ref{subsection:outcomes},
  \[
  |\mu_{b^k}(Z_x) - \mu_b(Z_x)| \underset{k \shortrightarrow
    \infty}{\longrightarrow} 0.
    \]
    In other words the sequence of measures
  $\mu_{b^k}$ \emph{strongly converges} to $\mu_b$. Since, for every player $i$,
  the function $\pi_i$ is measurable and bounded, it follows  
  that
  \[
  \int_{\Sigma} \pi_i\ d\mu_{b^k} \underset{k \shortrightarrow
    \infty}{\longrightarrow} \int_{\Sigma} \pi_i\ d\mu_b.
    \]
    Therefore,   $\Pi_i(b^k) \underset{k \shortrightarrow \infty}{\longrightarrow}
  \Pi_i(b)$, and thus the expected payoff function $\Pi$ is continuous.
\end{proof}

\subsection{Proof of Theorem~\ref{cor:mainMIS}}

%\begin{theorem}
%  Every LCL game has a trembling-hand perfect equilibrium.
%\end{theorem}
%
%\begin{proof}
  The above claims show that every LCL game satisfies the requirements of
  Lemma~\ref{theo:maingeneral}, and therefore has a symmetric trembling-hand perfect
  equilibrium. \qed

%%%%%%%%%%%%%%%%%%%%%%%%%%%%%%%%%%%%%%%%%%%%%
\section{Conclusion}
%%%%%%%%%%%%%%%%%%%%%%%%%%%%%%%%%%%%%%%%%%%%%

The objective of this paper is to address the issue of selfish behaviors in the context of distributed network computing.  Theorem~\ref{cor:mainMIS} establishes that distributed algorithms coping with selfishness do exist for a large fraction of locally checkable labeling tasks. This result is the necessary first step toward designing distributed algorithms in which the players have no incentive to deviate from their given instructions. One direction for further research consists in looking for efficient (centralized)  algorithms for computing the equilibria, as  as well as (distributed) protocols leading the nodes to automatically adopt the desired behavior, i.e., the one with highest social benefit, either in term of converging time or quality of the solution, potentially using incentives.  Another question for further studies is to measure the efficiency of  trembling-hand Nash equilibria for LCL games in term of round-complexity. Apart from a few special cases, like, e.g., complete networks, these issues are challenging problems, but quite rewarding both theoretically and practically. 
Finally, from a broader perspective, it would be desirable to characterize efficient distributed algorithms in terms of a suitable notion of equilibrium in (extensive form) games.

%%%%%%%%%%%%%%%%%%%%%%%%%%%%%%%%%%%%%%%%%%%%%
\newpage
\bibliographystyle{plain}
\bibliography{biblio}
%%%%%%%%%%%%%%%%%%%%%%%%%%%%%%%%%%%%%%%%%%%%%

\end{document}